\newif\ifconfver
\confverfalse      

\newif\ifcutshort      
\cutshorttrue

\newif\ifcutshortlvltwo  
\cutshortlvltwofalse

    \documentclass[journal]{IEEEtran}

\usepackage[dvips]{epsfig}
\usepackage[dvips]{graphicx}
\usepackage{boxedminipage}
\usepackage{color}
\usepackage{amsfonts}
\usepackage{amssymb}
\usepackage{amsmath}
\usepackage{graphics,amsmath,amsfonts,amssymb,epsfig,latexsym,amsfonts,graphicx,amssymb}
\usepackage{setspace}
\usepackage{cite,subfigure}
\usepackage{wrapfig,floatrow}

\newfloatcommand{capbtabbox}{table}[][\FBwidth]

\usepackage{blindtext}

\newcommand{\trace}{{\mbox{\textrm{\rm Tr}}}}
\newcommand{\rank}{{\mbox{\textrm{Rank}}}}

\newcommand{\bfSigma}{{\mbox{\boldmath $\Sigma$}}}

\newcommand{\st}{{\rm s.t.}}

\newcommand{\by}{\mathbf{y}}

\newcommand{\bx}{\mathbf{x}}
\newcommand{\bs}{\mathbf{s}}
\newcommand{\bL}{\mathbf{L}}
\newcommand{\bJ}{\mathbf{J}}
\newcommand{\bE}{\mathbf{E}}

\newcommand{\bC}{\mathbf{C}}
\newcommand{\bD}{\mathbf{D}}
\newcommand{\bM}{\mathbf{M}}
\newcommand{\bR}{\mathbf{R}}
\newcommand{\bH}{\mathbf{H}}
\newcommand{\bI}{\mathbf{I}}
\newcommand{\bU}{\mathbf{U}}
\newcommand{\bP}{\mathbf{P}}
\newcommand{\bQ}{\mathbf{Q}}

\newcommand{\bW}{\mathbf{W}}
\newcommand{\bV}{\mathbf{V}}

\newcommand{\bX}{\mathbf{X}}
\newcommand{\bY}{\mathbf{Y}}

\newcommand{\bd}{\mathbf{d}}

\newcommand{\bS}{\mathbf{S}}

\newcommand{\bhV}{\widehat{\mathbf{V}}}

\newcommand{\bhJ}{\widehat{\mathbf{J}}}
\newcommand{\bhS}{\widehat{\mathbf{S}}}

\newcommand{\bhC}{\widehat{\mathbf{C}}}
\newcommand{\bhU}{\widehat{\mathbf{U}}}
\newcommand{\bhE}{\widehat{\mathbf{E}}}
\newcommand{\btR}{\widetilde{\mathbf{R}}}
\newcommand{\btV}{\widetilde{\mathbf{V}}}
\newcommand{\hc}{\widehat{c}}
\newcommand{\tg}{\widetilde{g}}

\newcommand{\cI}{\mathcal{I}}
\newcommand{\cK}{\mathcal{K}}
\newcommand{\cQ}{\mathcal{Q}}
\newcommand{\cV}{\mathcal{V}}
\newcommand{\cS}{\mathcal{S}}

\newtheorem{lemma}{Lemma}

\newtheorem{thm}{Theorem}

\newtheorem{coro}{Corollary}

\newtheorem{remark}{Remark}

\begin{document}
\title{\huge Decomposition by Successive Convex Approximation:
A Unifying Approach for Linear Transceiver Design in Heterogeneous Networks}
\author{\normalsize Mingyi Hong, Qiang Li, Ya-Feng Liu \thanks{M. Hong is with the Department of Industrial and Manufacturing Systems Engineering, Iowa State University, Ames, IA, USA (email: \texttt{mingyi@iastate.edu}). Q.
Li is with the School of Communication and Information Engineering, University of Electronic Science and
Technology of China, Chengdu, China (email: \texttt{lq@uestc.edu.cn}). Y.-F. Liu is with the State Key Laboratory of Scientific and Engineering Computing, Institute of Computational Mathematics and Scientific/Engineering Computing, Academy of Mathematics and Systems
Science, Chinese Academy of Sciences, Beijing, 100190, China (email: \texttt{yafliu@lsec.cc.ac.cn}). M. Hong is supported in part by NSF under Grant CCF-1526078. Q. Li is supported in part by the National Natural Science Foundation of China under Grant 61401073 and 61531009. Y.-F. Liu is partially supported by the National Natural Science Foundation of China, Grants 11301516 and 11331012.}} \maketitle
\begin{abstract}
We study the downlink linear precoder design problem in a multi-cell dense heterogeneous network (HetNet). The problem is formulated as a general sum-utility maximization (SUM) problem, which includes as special cases many practical precoder design problems such as multi-cell coordinated linear precoding, full and partial per-cell coordinated multi-point transmission, zero-forcing precoding and joint BS
clustering and beamforming/precoding. The SUM problem is difficult due to its
non-convexity and the tight coupling of the users' precoders. In this paper we propose a novel convex approximation technique to approximate the original problem by a series of convex subproblems, each of which decomposes across all the cells.
The convexity of the subproblems allows for efficient computation,
while their decomposability leads to distributed implementation. {Our approach hinges upon the identification of certain key convexity properties of the sum-utility objective, which allows us to transform the problem into a form that can be solved using a popular algorithmic framework called BSUM (Block Successive Upper-Bound Minimization).} Simulation experiments show that
the proposed framework is effective for solving interference
management problems in large HetNet.
\end{abstract}
%

\section{Introduction}

Heterogeneous network (HetNet) has recently emerged as a promising
wireless network architecture \cite{3gpp09}. In HetNet, the upper tier
high-power Base Stations (BSs) such as Macro BSs provide per-cell interference
management and blanket coverage, while  the lower tier
low-power access points such as micro/pico/femto BSs are
densely deployed to provide capacity extension. This new paradigm of network design brings the
transmitters and receivers closer to each other, enabling
high link quality with low transmission power \cite{damnjanovic11}.

Due to the presence of a large number of potential interfering nodes, one of the key challenges in the design of the HetNet is to properly mitigate both the inter-cell and intra-cell multiuser interference. Interference management for the HetNet, or for general interfering
networks, has been under extensive research recently \cite{hong12survey,bjornson13}. An effective approach is to introduce appropriate coordination among the network nodes, either in the physical layer or in higher layers \cite{gesbert10}. For example, physical-layer coordination can take the form of
coordinated beamforming (CB) or joint processing (JP) \cite{gesbert10}. The CB coordinates the nodes in
the beamformer/precoder level \cite{Ho10, liu11MISO}, while the JP (a.k.a. the coordinated multi-point transmission, CoMP)
optimizes the transceivers assuming that the users' data are available at all the BSs.  In practice, dynamic combination of the two techniques is usually adopted for signaling overhead reduction \cite{zhang09, gesbert10}.

Mathematically, interference management is commonly formulated as some form of
sum-utility maximization (SUM) problem \cite{hong12survey}. The utility functions, when
chosen properly, can well balance the network spectrum efficiency
and the fairness level among the users. Unfortunately, for a large class of popular utility functions,
the associated optimization problems are difficult to solve (except
for a few special cases, see \cite{luo08a, liu11MISO,
liu13maxminTSP, liu13spl}). Therefore, many
low-complexity algorithms have been recently
developed. Often, the key in finding a good algorithm is to recognize certain convexity and
decomposability of the SUM problem, as the former leads to efficient computation, while
the latter is critical for distributed implementation \cite{Palomar06}.

The hidden convexity structures of various SUM problems have been extensively investigated in the past.
For example, reference \cite{Ye03} is the first to recognize the now well-known concave-convex property for a MIMO
interfering channel (IC), i.e., each users' achievable rate is {concave} in its own transmit covariance while being convex in all other interfering users' transmit covariances. Similar concave-convex properties have since been leveraged heavily to design resource allocation algorithms in different network settings \cite{Shi:2009,Tsiaflakis10, hong12_icassp, kim11, Ho10, papand09, liu11MISO, scutari13decomposition, Scutari12pricing, Weeraddana:2012:WSM:2432407}.  However, most of the schemes cited above do not decompose well across the nodes: for the schemes based on the convex-concave structure, the convexification procedures can only be done one node at each time, thus only {\it a single} node can
update its transmit strategy in each iteration; for the algorithm proposed in \cite{papand09}, the convex subproblem is still coupled
among all the users. Moreover, most of these algorithms are designed to handle {\it peer-to-peer networks}, with each transmitter transmiting to a single receiver, and/or with each receiver receiving from a single transmitter. Hence they are not directly applicable to the
HetNet setting where each BS or user can communicate with multiple nodes. Further, all the schemes mentioned above deal with smooth SUM problems, therefore they are not applicable to formulations that require nonsmooth penalizations \cite{hong12sparse}. Classical subgradient (SG) method \cite{bertsekas99} can potentially be used to solve such nonsmooth optimization, but in practice the SG can be very slow and its performance depends critically on the choice of the stepsizes.

Nevertheless, decomposability structure of interference management problems is highly desirable. When judiciously exploited, it
leads to efficient distributed implementation. This fact has long been recognized and leveraged in other important large-scale network optimization problems such as the network utility maximization (NUM) problems; see \cite{Palomar06}. However, unlike the NUM, the nodes
in the interfering networks are tightly coupled in a nonlinear manner through multi-user interference. Therefore even in simple peer-to-peer interfering networks, decomposability structure is difficult to come by. To address this issue, reference \cite{shi11WMMSE_TSP} proposes a weighted minimum mean square error (WMMSE) algorithm to solve the SUM problem in interfering broadcast channels (IBC). The most desirable feature of WMMSE is that each of its steps completely decouples among the interfering BSs. Stationary solutions of the SUM problem can be obtained via solving three subproblems alternatingly.  {Recently, there are a few works \cite{scutari13decomposition, Scutari12pricing, meisam14nips,Facchinei15} developing general parallel schemes for distributed optimization in multi-agent systems, which cover interfering wireless networks as a special case. For example, in \cite{scutari13decomposition} an interesting parallelization technique based on partial linearization is proposed. Generally speaking, there are subtle differences between parallelization and decomposition. The {parallelization techniques} such as those developed in \cite{scutari13decomposition, Scutari12pricing, meisam14nips,Facchinei15} gear toward solving a general multi-block problem a few blocks a time, but all requiring careful stepsize tuning. The decomposition techniques such as those used in \cite{Palomar06} and the ones to be discussed in this paper are more problem specific: they often find an alternative representation of the problem in an iterative manner, so that the variables automatically become independent (hence the name decomposable). We will elaborate more on the distinctions between these two types of approaches in the following presentation.} 

In this work, we propose to achieve decomposability by means of {\it successive convex approximation}. {The key novelty of our work is the  identification of an interesting hidden convexity for a wide range of sum-utility maximization problems, which allows us to interpret and generalize many existing algorithms under the recently developed BSUM (Block Successive Upper-Bound Minimization) framework \cite{Razaviyayn12SUM,hong15busmm_spm}}. By using BSUM, the original non-convex problem is approximated by a series of convex subproblems, each of which is completely decoupled among the network nodes. Based on different ways of solving the convex
subproblems, two low-complexity algorithms are proposed, each having wide applicability in interference management problems.

 The rest of the paper is organized as follows. In Section \ref{secSys}, we describe the system model and problem formulation. In Section \ref{secSCA}, we present a key convexity structure which leads to two general SCA algorithms. In Section \ref{secCustomizedAlgorithm}, the proposed algorithms are specialized to various interference management scenarios. Numerical results are given in Section \ref{secSimulation}, followed by concluding remarks.

 {\it Notations}: For a symmetric matrix
$\mathbf{X}$, $\mathbf{X}\succeq 0$ signifies that $\mathbf{X}$ is Hermitian
positive semi-definite. We use $\trace[\mathbf{X}]$, $|\mathbf{X}|$,
$\mathbf{X}^H$, $\rho(\mathbf{X})$, $\|\bX\|_F$ and
$\rank(\mathbf{X})$ to denote the trace, determinant, Hermitian,
spectral radius, Frobenius norm, and the rank of a matrix,
respectively. We use $\langle\cdot, \cdot\rangle$ to denote the inner product operation. The $(m,n)$-th element of a matrix $\bX$ is denoted by
$\bX[m,n]$. We use $\mathbf{I}_n$ to denote an $n\times n$ identity
matrix. Moreover, we let $\mathbb{R}^{N\times M}$ and
$\mathbb{C}^{N\times M}$ denote the set of real and complex $N\times
M$ matrices, and use $\mathbb{S}^{N}$, $\mathbb{S}^{N}_{+}$,
$\mathbb{S}^{N}_{++}$  to denote the set of $N\times N$ Hermitian,
Hermitian positive semi-definite and Hermitian positive definite
matrices, respectively. A list of other main notations are given in the following table.
\begin{table*}[htb]
\caption{ {A List of Notations} }
\begin{center}
{
\begin{tabular}{|c |c |c|c|}
\hline
$K$& \# of cells & $Q_k$& \#  of BSs in cell $k$ \\
\hline
$i_k$ & $i$-th user in $k$-th cell & $q_k$ & $q$-th BS in $k$-th cell\\
\hline
 $\mathbf{H}^{q_\ell}_{i_k}$ & channel matrix between the $q_{\ell}$-th BS and
$i_k$-th user & $I_k$ & \# of users in cell $k$\\
 \hline
 $\mathbf{H}^{\ell}_{i_k}$ & channel matrix between the all BSs in cell $\ell$ and
$i_k$-th user  & $d$& \# of data streams for each user\\
\hline
$\bV^{q_k}_{i_k}$ & precoder used by $q_{k}$-th BS for
$i_k$-th user & $\mathbf{U}_{i_k}$ & The receiver used by $i_k$-th user\\
\hline
$\bV_{i_k}$ & precoder used by all BSs in cell $k$ for
$i_k$-th user & $\mathbf{E}_{i_k}$ & The MSE matrix for $i_k$-th user\\
\hline
$\bC_{i_k}$ ($\bY_{i_k}$)& total received signal (interference) covariance matrix of user $i_k$ & $R_{i_k}$& transmit rate for user $i_k$ \\
\hline
\end{tabular} } \label{tableSymbols}
\end{center}
\vspace*{-0.2cm}
\end{table*}

\section{System Model and Problem Formulation}\label{secSys}
\subsection{Problem Description}
We consider a downlink multi-cell HetNet which consists of a set $\mathcal{K}:=\{1,\cdots,K\}$ of cells. Within each cell $k$ there is a set of $\mathcal{Q}_k:=\{1,\cdots,Q_k\}$ distributed BSs such as macro/micro/pico BSs which provide service to the users. Assume that in each cell $k$, there is a low-latency backhaul network connecting the set of BSs $\mathcal{Q}_k$ to a central controller (usually the macro
BS), who makes the resource allocation decisions for all BSs within the cell. The central controller has access to the data signals of all the users in its cell. Let
$\mathcal{I}_k:=\{1,\cdots, I_k\}$ denote the users associated with cell $k$. Each of the users $i_k\in\mathcal{I}_k$ is served jointly by a subset of BSs in $\mathcal{Q}_k$. Let $\mathcal{I}$ and $\cQ$ denote the set of all the users and all the BSs, respectively, and let $I=|\cI|$ and $Q=|\cQ|$. Assume that each BS has $M$ transmit antennas, and each user has $N$ receive antennas. Let $\mathbf{H}^{q_\ell}_{i_k}\in\mathbb{C}^{N\times M}$ denote the channel matrix between the $q$-th BS in the $\ell$-th cell and the
$i$-th user in the $k$-th cell. Similarly, we use $\mathbf{H}^{\ell}_{i_k}$ to denote the channel matrix between all the BSs in the $\ell$-th cell to the user $i_k$, i.e., $\mathbf{H}^{\ell}_{i_k}:=\big\{\mathbf{H}^{q_\ell}_{i_k}\big\}_{q_\ell\in\mathcal{Q}_\ell}\in
\mathbb{C}^{N\times MQ_\ell}$.

{Suppose that $d\le \min\{M,N\}$ data streams are transmitted to user $i_k$ (for notational simplicity we assume that all users have the same number of streams).} Let $\bV^{q_k}_{i_k}\in\mathbb{C}^{M\times d}$ denote the transmit
precoder that BS $q_k$ uses to transmit data
$\bs_{i_k}\in\mathbb{C}^{d}$ to user $i_k$. Define{
$$\bV_{i_k}:=\big\{\bV^{q_k}_{i_k}\big\}_{q_k\in\mathcal{Q}_k}, \;
\bV^{q_k}:=\big\{\bV^{q_k}_{i_k}\big\}_{i_k\in\mathcal{I}_k},\;  \bV^{k}:=\big\{\bV_{i_k}\big\}_{i_k\in\mathcal{I}_k}$$}
\!\!as the collection of all precoders intended for user $i_k$,
all precoders belong to BS $q_k$, and all precoders in cell $k$, respectively. Let
$\bV:=\{\bV_{i_k}\}_{i_k\in\mathcal{I}}$. Further assume that all the BSs in cell $k$ form a single virtual BS, and they jointly
transmit to user $i_k\in\cI_k$. Then $\bV_{i_k}$ can be viewed as the {\it virtual precoder} for user $i_k$. Using these definitions, we can express the transmitted signal of BS $q_k$ and the combined transmitted signal for all the BSs in cell $k$ as:
{ \begin{align}
{\bx}^{q_k}=\sum_{i_k\in\mathcal{I}_k}\bV^{q_k}_{i_k}\bs_{i_k}\in
\mathbb{C}^{M\times 1},\quad\quad
{\bx}^{k}=\sum_{i_k\in\mathcal{I}_k}\bV_{i_k}\bs_{i_k}\in\mathbb{C}^{M
Q_k\times 1}\nonumber.
\end{align}}
The received signal $\by_{i_k}\in\mathbb{C}^{N\times 1}$ of user
$i_k$ is{\small
\begin{align}
\by_{i_k}&=
\bH^{k}_{i_k}\bV_{i_k}\bs_{i_k}+\sum_{j_k\in\cI_k\setminus
i_k}\bH^{k}_{i_k}\bV_{j_k}\bs_{j_k}+{\sum_{\ell\ne
k}\sum_{j_\ell\in\mathcal{I}_\ell}\bH^{\ell}_{i_k}\bV_{j_\ell}\bs_{j_\ell}}+\mathbf{z}_{i_k}\nonumber
\end{align}}
\!\!where $\mathbf{z}_{i_k}\in\mathbb{C}^{N\times 1}$ is the complex Gaussian noise with distribution
$\mathcal{CN}(0,\sigma^2\mathbf{I}_N)$. Let $\bU_{i_k}\in\mathbb{C}^{N\times d}$ denote the
receiver used by user $i_k$ to decode the intended signal. Then the
estimated signal for user $i_k$ is:
$\widehat{\bs}_{i_k}=\bU^H_{i_k}\by_{i_k}$. The mean square error
(MSE) for user $i_k$ can be written as{
\begin{align}
\begin{split}
\bE_{i_k}&:=
\mathbb{E}[(\bs_{i_k}-\widehat{\bs}_{i_k})({\bs_{i_k}-\widehat{\bs}_{i_k}})^H]\\
&=(\bI_{d}-\bU^H_{i_k}\bH^{k}_{i_k}\bV_{i_k})(\bI_{d}-\bU^H_{i_k}\bH^{k}_{i_k}\bV_{i_k})^H\\
&\quad +\sum_{(\ell,j)\ne
(k,i)}\bU^H_{i_k}\bH^\ell_{i_k}\bV_{j_\ell}\bV^H_{j_\ell}(\bH^\ell_{i_k})^H\bU_{i_k}+\sigma^2_{i_k}\bU^H_{i_k}\bU_{i_k}.\label{eqMSE}
\end{split}
\end{align}}\hspace{-0.2cm}
The Minimum MSE (MMSE) receiver is given by \cite{verdu98}{
\begin{align}
\bU^{\rm{mmse}}_{i_k}&=\bigg(\sum_{(\ell,j)}\bH^{\ell}_{i_k}\bV_{j_\ell}\bV^H_{j_\ell}(\bH^{\ell}_{i_k})^H
+\sigma^2\bI_N\bigg)^{-1}\bH^{k}_{i_k}\bV_{i_k}\nonumber\\
&:= \bC^{-1}_{i_k}\bH^{k}_{i_k}\bV_{i_k}\label{eqDefCCapital}
\end{align}}
\!\!where $\bC_{i_k}\in\mathbb{S}^{N}_{++}$ is user $i_k$'s
received signal covariance matrix. Plugging \eqref{eqDefCCapital} into \eqref{eqMSE} we obtain
{
\begin{align}
\bE^{\rm
mmse}_{i_k}=\bI_{d}-\bV^H_{i_k}(\bH^k_{i_k})^H\bC_{i_k}^{-1}\bH^{k}_{i_k}\bV_{i_k}\succeq
0.\label{eqMMSE}
\end{align}}
Clearly we also have $\bI_{d}-\bE^{\rm mmse}_{i_k}\succeq 0$.

Let us assume that Gaussian signaling is used and the interference
is treated as noise. The achievable rate for user $i_k$ is given by
\cite{cover05}{
\begin{align}
R_{i_k}&=\log\Bigg|\mathbf{I}_{N}+\mathbf{H}^k_{i_k}\bV_{i_k}\bV^H_{i_k}({\bH}^k_{i_k})^H
\mathbf{Y}_{i_k}^{-1}\Bigg|=-\log\left|\mathbf{E}^{\rm mmse}_{i_k}\right| \label{eqRateMMSE}
\end{align}}
\!\!where we have defined the matrix $\mathbf{Y}_{i_k}:=\sum_{(\ell,j)\ne (k,i)}
\mathbf{H}^\ell_{i_k}\bV_{j_\ell}\bV^H_{j_\ell}({\bH}^{\ell}_{i_k})^H+\sigma^2\mathbf{I}_N$ as the received interference covariance matrix for user $i_k$; the last equality is a well-known relationship between the rate and the MMSE matrix (see e.g., \cite{christensen08}).  We will occasionally use the notations
$R_{i_k}(\bV)$, $\bC_{i_k}(\bV)$ and $\bE^{\rm mmse}_{i_k}(\bV)$ to make their dependencies on $\bV$ explicit.

Let $f_{i_k}(\bullet): \mathbb{R}_{+}\to\mathbb{R}$ denote the utility
function of user $i_k$'s data rate. {Let $s^{q_k}_{i_k}(\bullet)$ denote a penalty term for $\bV^{q_k}_{i_k}$, which is useful in inducing certain structures on the precoders such as sparsity or group sparsity \cite{hong12sparse}.} We assume that these functions satisfy the following assumptions:

\noindent{\bf Assumption A}.
\begin{description}
\item [\bf A-1)] $f_{i_k}(x)$ is a concave non-decreasing function in $x$ for all $x\ge 0$;
\item [\bf A-2)]  $f_{i_k}(-\log(|\mathbf{X}|))$ is convex in
$\mathbf{X}$, for all $\bI\succeq\mathbf{X}\succeq 0$;
\item [\bf A-3)] $f_{i_k}(x)$ is continuously differentiable;
\item [\bf A-4)] $s^{q_k}_{i_k}(\bV^{q_k}_{i_k})$ is a convex, continuous,
but possibly nonsmooth function.
\end{description}
Note that this family of utility functions includes well-known
utilities such as the weighted sum rate, the geometric mean of one
plus rate and the harmonic mean rate utility functions (see
\cite{shi11WMMSE_TSP}). 

We consider the general system-level sum utilities
maximization problem given below{
\begin{align}
\max_{\bV}&\quad u(\bV):=f(\bV)-s(\bV)\label{problemSyS}\tag{$\rm P_{\rm SYSTEM}$} \\
\st&\quad f(\bV):=\sum_{k\in\mathcal{K}}\sum_{i_k\in{\mathcal{I}_k}}f_{i_k}(R_{i_k}(\bV))\nonumber\\
&\quad
s(\bV):=\sum_{k\in\cK}\sum_{i_k\in\cI_k}\sum_{q_k\in\cQ_k}s^{q_k}_{i_k}(\bV^{q_k}_{i_k})\nonumber\\
 &\quad \bV^{q_k}\in\mathcal{V}^{q_k}, \
\forall~q_k\in\mathcal{Q},\nonumber\\
&\quad \bV^{k}\in\mathcal{V}^{k}, \
\forall~k\in\mathcal{K}\nonumber
\end{align}}
\!\!where $\mathcal{V}^{q_k}$ and $\mathcal{V}^{k}$ are some
feasible sets for $\bV^{q_k}$ and $\bV^{k}$, respectively. Let
$\mathcal{V}$ denote the feasible set for $\bV$.  In the next subsection we provide a few transceiver design problems arising in HetNet that are covered by the model \rm $(P_{\rm SYSTEM})$. 

\vspace{-0.3cm}
\subsection{Applications}
\subsubsection {\bf MIMO IBC/IMAC/IC channels with inter-BS CB \cite{shi2009pricingmimo,liu11MISO,christensen08,venturino10,kim11}}
Consider a network with $|\cQ_k|=1$ $\forall~k$, i.e., each cell has a single BS. In this case $s(\bV)\equiv 0$, and the constraint set $\mathcal{V}^{q_k}$ is the same as $\mathcal{V}^{k}$, which is given by the following sum-power constrained set ($\bar{P}_k$ denotes the power budget for cell $k$){
\begin{align}
\mathcal{V}^k=\bigg\{\bV^k:\sum_{i_k\in\mathcal{I}_k}\trace[\bV_{i_k}\bV_{i_k}^H]\le
\bar{P}_k\bigg\}\label{eqSumPower}.
\end{align}}

\vspace{-0.3cm}
\subsubsection  {\bf Multicell MIMO with intra-cell CoMP and inter-cell CB  \cite{Ng10, hong12sparse}}
The BSs in different cells cooperate in the precoder level, while those in the same cell perform JP. Each BS $q_k$ has a separate power budget, denoted as $\bar{P}^{q_k}$: {
\begin{align}
\mathcal{V}^{q_k}=\bigg\{\bV^{q_k}:\sum_{i_k\in\mathcal{I}_k}\trace[\bV^{q_k}_{i_k}(\bV^{q_k}_{i_k})^H]\le
\bar{P}^{q_k}\bigg\}\label{eqPerBSPower}.
\end{align}}
\!\!This model generalizes those for the IBC model (e.g., \cite{shi11WMMSE_TSP, venturino10}), in that the
sum-power constraint \eqref{eqSumPower} is replaced by the set of {\it per-group of antennas} power constraints.

\subsubsection  {\bf Multicell MIMO with intra-cell partial CoMP and inter-cell
CB\cite{hong12sparse, kim12}}   A practical alternative to full CoMP is to implement a {\it partial CoMP} strategy, where each user is served by only a few BSs in each cell. In this case, the BSs in the same cell are grouped into different (possibly overlapping) clusters with small sizes, within which they fully cooperate for joint transmission. To jointly design the BS clusters and the precoders, one needs to properly utilize the penalty term $s(\bV)$. The requirement that each user is served only by {\it a small number} of BSs translates to certain {\it block sparse} structure of $\bV_{i_k}=\{\bV^{q_k}_{i_k}\}_{q_k\in\cQ_k}$, see \cite{hong12sparse, hong12_asilomar, dai14sparse, liao13admm}. To induce such block sparsity, the penalty term can take the following form {
\begin{align}
s_{i_k}^{q_k}(\bV_{i_k}^{q_k})=\gamma^{q_k}_{i_k}\|\bV^{q_k}_{i_k}\|_F\label{eqL2Penalization},
\end{align}}
\!\!with $\gamma^{q_k}_{i_k}\ge 0$ being some constant. See \cite{hong12sparse, hong12_asilomar, dai14sparse, liao13admm} and the references therein for motivation of
using \eqref{eqL2Penalization}. It is worth noting that most of the existing decomposition based algorithms \cite{papand09, liu11MISO, scutari13decomposition, Scutari12pricing} cannot handle such nonsmooth problem.

\subsubsection  {\bf Multicell MIMO network with intra-cell ZF and inter-cell
CB \cite{zhang09, zhang10JSAC}} In this case all the BSs in the same cell
jointly perform the ZF precoding (a.k.a. the block-diagonalization precoding) \cite{Spencer04}, while the BSs in different cells perform CB. The feasible sets are given as{\small
\begin{align}
\mathcal{V}^{k}&=\big\{\bV^{k}:\bH^k_{j_k}\bV_{i_k}(\bV_{i_k})^H(\bH^k_{j_k})^H={\bf
0},\forall~j_k\ne
i_k,\ j_k,i_k\in\mathcal{I}_k\big\}\nonumber\\
\mathcal{V}^{q_k}&=\left\{\bV^{q_k}:\sum_{i_k\in\mathcal{I}_k}\trace[\bV^{q_k}_{i_k}(\bV^{q_k}_{i_k})^H]\le
\bar{P}^{q_k},
\forall~q_k\in\mathcal{Q}_k\right\}.\nonumber
\end{align}}
\!\!Despite the wide applicability of \eqref{problemSyS}, solving it to its global optimality is often difficult (i.e., NP-hard) \cite{luo08a,
liu13maxminTSP, liu11MISO,liu13spl}. The NP-hardness of the problem indicates that it is even unlikely to find an equivalent convex reformulation for it (unless P=NP). Therefore the best that one can do is to seek efficient algorithms that provide approximately optimal solutions. Unfortunately, the tight coupling of the variables
$\{\bV^k\}_{k=1}^{K}$ in the objective often makes such task challenging. We mention that existing general optimization frameworks are not directly applicable to solve \eqref{problemSyS}. For example the partial linearization approach proposed in \cite{scutari13decomposition, Scutari12pricing} cannot handle nonsmooth terms in the objective. The algorithms developed in \cite{meisam14nips, Facchinei15} are parallelization schemes that can potentially handle non-convex and nonsmooth objective. However when applied to \eqref{problemSyS}, it requires that each subproblem (say for optimizing BS $q_k$'s precoder $\bV^{q_k}$) or its approximation is easy to solve.  It is not clear how to construct an effective approximation for \eqref{problemSyS} here. Further a careful stepsize tuning is needed to enable parallel update.

\vspace{-0.2cm}
\section{A Successive Convex Approximation
Approach}\label{secSCA}

In this section, we present our main approach for computing a
high quality solution for the general problem \eqref{problemSyS}. Our approach is to
solve, possibly in an inexact manner, a series of convex subproblems each of which is a
local approximation of \eqref{problemSyS}.

{We first highlight the main steps of our subsequent development. Our first step is to derive an approximation of the objective function $u(\bV)=f(\bV)-s(\bV)$ at a given feasible point $\bhV$, denoted as $h^{\beta}(\bV; \bhV)-s(\bV)$, for any feasible $\bV$ ($\beta> 0$ is some constant parameter). The construction of the function $h^{\beta}(\bV; \bhV)$ requires a novel two-stage convex approximation. It turns out that for every fixed input argument $\bhV$, our approximation is a global lower bound of $u(\bV)$, and it is easy to optimize with respect to $\bV$. More importantly, this lower bound decouples completely among the variables, that is, we have the following decomposition structure
{{
\begin{align}\label{eq:decomposable}
h^{\beta}(\bV; \bhV)=\sum_{i_k\in\mathcal{I}}g^{\beta}_{i_k}(\bV_{i_k};\bhV),
\end{align}}}
\!for some $g^{\beta}_{i_k}(\bullet;\bhV)$ which is only a function of $\bV_{i_k}$. Our second step constructs algorithms that iteratively optimize  $h^{\beta}(\bV; \bhV)-s(\bV)$ with respect to $\bV$, while taking the previous iterate $\bhV$ as the input argument. Thanks to the decomposability structure \eqref{eq:decomposable}, the resulting computation can be completely distributed into each cell.
}

\vspace{-0.2cm}
\subsection{A Local Approximation for \eqref{problemSyS}}\label{subLowerBound}
We begin by deriving a simple local approximation for the individual
users' utility function $f_{i_k}(\bullet)$ using the convexity assumption A-2). Let $\widehat{\bV}\in\cV$ denote a
feasible solution to problem \eqref{problemSyS}. Let
$\widehat{\bE}_{i_k}:=\bE^{\rm mmse}_{i_k}(\widehat{\bV})$
denote the MMSE evaluated at $\widehat{\bV}$.  Leveraging \eqref{eqRateMMSE}, we can express $f_{i_k}(\bullet)$ as a function of
$\bE^{\rm mmse}_{i_k}$ only:{\small
\begin{align}
&f_{i_k}(R_{i_k}(\bV))=f_{i_k}(-\log|\bE^{\rm mmse}_{i_k}|)\nonumber\\
&\stackrel{\rm (i)}\ge f_{i_k}(-\log|\widehat{\bE}_{i_k}|)-\frac{\partial
f_{i_k}(x)}{\partial
x}\bigg|_{x=-\log|\widehat{\bE}_{i_k}|}\hspace{-0.5cm}\trace\left[\widehat{\bE}_{i_k}^{-1}(\bE^{\rm mmse}_{i_k}-\widehat{\bE}_{i_k})\right]\nonumber\\
&:= \widehat{a}_{i_k}-\hc_{i_k}\trace\left[\widehat{\bE}_{i_k}^{-1}(\bE^{\rm mmse}_{i_k}-\widehat{\bE}_{i_k})\right]\nonumber\\
&:=\bar{h}_{i_k}\left(\bE^{\rm mmse}_{i_k}, \bV_{i_k};
\widehat{\bE}_{i_k}, \bhV_{i_k}\right)\label{eqFirstApproximation}
\end{align}}
{\!\! where $\hc_{i_k}$ and $\widehat{a}_{i_k}$ are two constants not related to either $\bE^{\rm mmse}_{i_k}$ or $\bV$; the inequality in ${\rm (i)}$ is due to the fact that the convex function is always lower bounded by its local linear approximation. Note that $\hc_{i_k}\ge 0$ due
to the non-decreasing property of $f_{i_k}(\bullet)$'s assumed in A-1).}

In fact, from the above derivation it is readily seen that $\bar{h}_{i_k}(\bE^{\rm mmse}_{i_k}, \bV_{i_k};
\widehat{\bE}_{i_k}, \bhV_{i_k})$ is a global lower bound of $f_{i_k}(\bullet)$ at the point $(\widehat{\bE}_{i_k}, \bhV_{i_k})$, i.e., for all feasible
$(\bE^{\rm mmse}_{i_k}, \bV_{i_k})$ and feasible $(\bhE_{i_k}, \bhV_{i_k})$ we have{
\begin{align*}
\bar{h}_{i_k}(\bE^{\rm mmse}_{i_k}, \bV_{i_k};
\bhE_{i_k}, \bhV_{i_k})&\le f_{i_k}(-\log|\bE^{\rm mmse}_{i_k}|)\nonumber\\
\bar{h}_{i_k}(\bhE_{i_k}, \bhV_{i_k}; \bhE_{i_k}, \bhV_{i_k})&=f_{i_k}(-\log|\bhE_{i_k}|).
\end{align*}}
\!\!Unfortunately, such approximation does not simplify the problem, as summing $\bar{h}_{i_k}(\bullet)$ over the users still results in a non-convex non-separable function with respect to $\bV$ (cf. \eqref{eqMMSE}).

Our next step performs a {\it second-stage approximation}, in which we further approximate $\bar{h}_{i_k}(\bullet)$ by a concave function of $\bV$. To this end, we need a key lemma that explores some hidden
convexity property of the function $\bar{h}_{i_k}(\bullet)$.

\begin{lemma}\label{lemmaConvexity}
\it The following function
{$$l_{i_k}(\bV_{i_k}, \bC_{i_k}):=\trace\left[\widehat{\bE}_{i_k}^{-1}
\bV_{i_k}^H(\bH^k_{i_k})^H\bC^{-1}_{i_k}\bH^k_{i_k}\bV_{i_k}\right]$$}
is jointly convex with respect to the pair of variables $(\bV_{i_k},
\bC_{i_k})$ in the feasible region $(\mathbb{C}^{MQ_k\times
d}, \mathbb{S}^{N}_{++})$.
\end{lemma}

{The proof is given in Appendix \ref{app:Lemma1}.} The result provided by Lemma \ref{lemmaConvexity} allows us to
further approximate the following function $\bar{h}_{i_k}(\bullet)$. To describe such second-stage approximation, we define
$\widehat{\bC}_{i_k}:=\bC_{i_k}(\widehat{\bV})$, and note {\small
\begin{align}\label{eq:second_layer}
&l_{i_k}(\bV_{i_k},
\bC_{i_k})=\trace\left[\widehat{\bE}_{i_k}^{-1}
\bV_{i_k}^H(\bH^k_{i_k})^H\bC^{-1}_{i_k}\bH^k_{i_k}\bV_{i_k}\right]\nonumber\\
&\stackrel{\rm (i)}\ge \trace\left[\widehat{\bE}_{i_k}^{-1}
\bhV_{i_k}^H(\bH^k_{i_k})^H\bhC^{-1}_{i_k}\bH^k_{i_k}\bhV_{i_k}\right]\nonumber\\
&\quad +\frac{d
l_{i_k}(\widehat{\bV}_{i_k}+t(\bV_{i_k}-\widehat{\bV}_{i_k}),
\widehat{\bC}_{i_k})}{d t}\bigg|_{t=0}\nonumber\\
&\quad +\frac{d
l_{i_k}(\widehat{\bV}_{i_k},
\widehat{\bC}_{i_k}+t(\bC_{i_k}-\widehat{\bC}_{i_k}))}{d
t}\bigg|_{t=0}\nonumber\\
&\stackrel{\rm (ii)}=\trace[\bhE_{i_k}^{-1}]-d+\trace\left[\widehat{\bE}_{i_k}^{-1}\widehat{\bV}^H_{i_k}(\bH^k_{i_k})^H\widehat{\bC}_{i_k}^{-1}
\bH^k_{i_k}(\bV_{i_k}-\widehat{\bV}_{i_k})\right]\nonumber\\
&\quad +\trace\left[\widehat{\bE}_{i_k}^{-1}(\bV_{i_k}-\widehat{\bV}_{i_k})^H(\bH^k_{i_k})^H\widehat{\bC}_{i_k}^{-1}
\bH^k_{i_k}\widehat{\bV}_{i_k}\right]\nonumber\\
&\quad-
\trace\left[\widehat{\bE}_{i_k}^{-1}\widehat{\bV}^H_{i_k}(\bH^k_{i_k})^H\widehat{\bC}_{i_k}^{-1}(\bC_{i_k}-\widehat{\bC}_{i_k})
\widehat{\bC}_{i_k}^{-1}\bH^k_{i_k}\widehat{\bV}_{i_k}\right]
\end{align}}
\!\!{where $\rm (i)$ is due to the convexity of $l_{i_k}(\bullet)$ derived in Lemma \ref{lemmaConvexity}; $\rm (ii)$ simply uses the definition of directional derivatives; see Appendix \ref{appDirectional} for detailed derivation.} The above inequality combined
with the definition in \eqref{eqFirstApproximation} yields:{\small
\begin{align}
&\bar{h}_{i_k}(\bE^{\rm mmse}_{i_k}, \bV_{i_k};
\bhE_{i_k}, \bhV_{i_k})\nonumber\\
&\stackrel{\rm
(i)}=\widehat{a}_{i_k}+\hc_{i_k}\trace\left[\widehat{\bE}_{i_k}^{-1}
\bV_{i_k}^H(\bH^k_{i_k})^H\bC^{-1}_{i_k}\bH^k_{i_k}\bV_{i_k}\right]+\hc_{i_k}d-\hc_{i_k}\trace[\bhE_{i_k}^{-1}]\nonumber\\
&\stackrel{\rm (ii)}\ge \widehat{a}_{i_k}+
\hc_{i_k}\trace\big[\widehat{\bE}_{i_k}^{-1}\widehat{\bV}^H_{i_k}(\bH^k_{i_k})^H\widehat{\bC}_{i_k}^{-1}
\bH^k_{i_k}(\bV_{i_k}-\widehat{\bV}_{i_k})\big]\nonumber\\
&\quad +\hc_{i_k}\trace\big[\widehat{\bE}_{i_k}^{-1}(\bV_{i_k}-\widehat{\bV}_{i_k})^H(\bH^k_{i_k})^H\widehat{\bC}_{i_k}^{-1}
\bH^k_{i_k}\widehat{\bV}_{i_k}\big]\nonumber\\
&\quad-
\hc_{i_k}\trace\left[\widehat{\bE}_{i_k}^{-1}\widehat{\bV}^H_{i_k}(\bH^k_{i_k})^H\widehat{\bC}_{i_k}^{-1}(\bC_{i_k}-\widehat{\bC}_{i_k})
\widehat{\bC}_{i_k}^{-1}\bH^k_{i_k}\widehat{\bV}_{i_k}\right]\nonumber\\
&\quad -\beta\left\|\bV_{i_k}-\bhV_{i_k}\right\|^2_F\nonumber\\
&\stackrel{\rm (iii)}=\tilde{a}_{i_k}+
\hc_{i_k}\trace\bigg[\widehat{\bE}_{i_k}^{-1}\bigg(\widehat{\bU}^H_{i_k}
\bH^k_{i_k}\bV_{i_k}+\bV_{i_k}^H(\bH^k_{i_k})^H\widehat{\bU}_{i_k}\nonumber\\
&\quad -\widehat{\bU}^H_{i_k}
\big(\sum_{(\ell,j)}\bH^{\ell}_{i_k}\bV_{j_\ell}(\bV_{j_\ell})^H(\bH^{\ell}_{i_k})^H\big)
\widehat{\bU}_{i_k}\bigg)\bigg]\nonumber\\
&\quad -\beta\trace\left[\bV_{i_k}\bV^H_{i_k}-\bV_{i_k}\bhV^H_{i_k}-\bhV_{i_k}\bV^H_{i_k}\right]:= h^{\beta}_{i_k}(\bV; \bhV)\label{eqDefh}
\end{align}}
\!\!where $\rm (i)$ is due to the definition of the MMSE matrix
\eqref{eqMMSE}; $\rm(ii)$ uses the inequality \eqref{eq:second_layer} and the fact that for any $\beta\ge 0$, the proximal term $-\beta\left\|\bV_{i_k}-\bhV_{i_k}\right\|^2_F$ is always nonnegative; in $\rm (iii)$, we have defined
$\widehat{\bU}_{i_k}:=\bhC^{-1}_{i_k}\bH^k_{i_k}\bhV_{i_k}$ and used the definition of $\bC_{i_k}$ in \eqref{eqDefCCapital};
$\tilde{a}_{i_k}$ collects all the terms that are not dependent
on $\bV$. It is interesting to observe that the newly defined quantity $\widehat{\bU}_{i_k}$ is in fact the MMSE receiver for user $i_k$ when the system precoder is given by $\bhV$ (cf. \eqref{eqDefCCapital}).

{We note that in the approximation \eqref{eqDefh} a proximal term $-\beta\left\|\bV_{i_k}-\bhV_{i_k}\right\|^2_F$ has been introduced, where $\beta\ge 0$ is some nonnegative constant. The reasons for introducing such term are twofold. First, it makes each approximation function $h^{\beta}_{i_k}$ strongly convex over $\bV_{i_k}$, therefore easier to optimize (with a unique solution). Second, it helps to construct an algorithm with provable convergence guarantee. The latter point will be clarified shortly in Remark \ref{remark:proximal}.

}

Combining \eqref{eqFirstApproximation} and \eqref{eqDefh}, we see that $h^{\beta}_{i_k}(\bV; \bhV)$ is a local approximation of $f_{i_k}(R_{i_k}(\bV))$, which satisfies the following two properties:{
\begin{align}
&{h}^{\beta}_{i_k}(\bhV; \bhV)=f_{i_k}(R_{i_k}(\bhV))\nonumber\\
&{h}^{\beta}_{i_k}(\bV; \bhV)\le
f_{i_k}(R_{i_k}(\bV)), \; \forall~\bhV, \bV\in \cV\label{eqHLowerBound}.
\end{align}}
\!\!Let us define  $h^{\beta}(\bV;
\bhV):=\sum_{i_k\in\mathcal{I}}{h}^{\beta}_{i_k}(\bV; \bhV)$. Then we readily have {\small
\begin{align}
\begin{split}\label{eqLowerBoundPropertyInequality}
&{h}^{\beta}(\bhV; \bhV)=f(\bhV), \quad\quad\quad
{h}^{\beta}(\bV; \bhV)\le f(\bV), \; \forall~\bhV, \bV\in \cV, \\
&{h}^{\beta}(\bhV; \bhV) - s(\bV)=u(\bhV),
{h}^{\beta}(\bV; \bhV)-s(\bhV)\le u(\bV), \; \forall~\bhV, \bV\in \cV.
\end{split}
\end{align}}
\!\!That is, the function $h^{\beta}(\bV; \bhV)$ is a locally tight, {\it
global} lower bound for the sum-utility function $f(\bV)$. A
direct consequence of this observation is that the function
$h^{\beta}(\bV;\bhV)-s(\bV)$ is a universal lower bound for $u(\bV)$.

Another key property of $h^{\beta}(\bV; \bhV)$ is that it is a {\it quadratic} function {\it separable} over all variables $\{\bV_{i_k}\}$. To see this, we can expand $h^{\beta}(\bV;
\bhV)$ to obtain \eqref{eqHExpanding},
\begin{table*}[t]
\normalsize{ \vspace{-0.2cm}
{{\small
\begin{align}
h^{\beta}(\bV; \bhV)&=\sum_{i_k\in\mathcal{I}}\Bigg(\tilde{a}_{i_k}
+\trace\Bigg[2\bigg(\hc_{i_k}\widehat{\bE}_{i_k}^{-1}\widehat{\bU}^H_{i_k}
\bH^k_{i_k}+\beta\bhV^H_{i_k}\bigg)\bV_{i_k}-\bV_{i_k}^H\bhJ^k
\bV_{i_k}\Bigg]\Bigg):=\sum_{i_k\in\mathcal{I}}g^{\beta}_{i_k}(\bV_{i_k};\bhV)\label{eqHExpanding}.
\end{align}}}
\rule{\linewidth}{0.2mm}}
\vspace{-0.8cm}
\end{table*}
where we have defined{\small
\begin{align}
\bhJ^k&:=\sum_{j_l\in\mathcal{I}}\hc_{j_l}(\bH^{k}_{j_l})^H\bhU_{j_l}\bhE_{j_l}^{-1}\bhU^H_{j_l}\bH^{k}_{j_l}+\beta\bI_{MQ_k}\in\mathbb{S}_{++}^{MQ_k}\label{eqDefJ}.
\end{align}}
Clearly the approximation function  $h^{\beta}(\bV;\bhV)$ consists of $|\cI|$ component functions, each of which is a function of a {\it single} variable $\bV_{i_k}$ (user $i_k$'s precoder).

Let us briefly summarize what we have so far. Our two-stage approximation not only convexifies the objective
function $u(\bV)$, but more importantly {\it decomposes} the
nonlinearly coupled objective $u(\bV)$. This property of the
sum-utility function $u(\bV)$ will be leveraged heavily for designing efficient distributed
algorithms.
{
\begin{remark}
{\it When $\beta=0$, the first-stage approximation expressed in \eqref{eqFirstApproximation} is well-known and has been used in existing works such as \cite{shi11WMMSE_TSP} \cite{christensen08}; also see \cite[Section VIII-A]{Razaviyayn12SUM} for discussion. However, none of those works has articulated the fact that there exists a separable quadratic function of $\bV$ in the form of \eqref{eqHExpanding} which serves as a global lower bound for the original system utility; cf. \eqref{eqLowerBoundPropertyInequality}. The latter property has been made clear in our derivation precisely due to the use of the {\it two-stage} approximation.}\hfill
$\square$
\end{remark}
}

{
\begin{remark}
{\it The main message here is that by performing the two-stage approximation, the resulting $h^{\beta}(\bV; \bhV)$, which can be viewed as an alternative representation of the original objective $u(\bV)$, completely decomposes across the optimization variable $\bV$. Such decomposability is intrinsic to the problem and independent of any algorithm to be used. Our method is therefore philosophically different from the general schemes discussed in \cite{scutari13decomposition, Scutari12pricing, meisam14nips,Facchinei15}, where {\it parallelization} is made possible by appropriate choices of algorithms and their parameters. \hfill
$\square$}
\end{remark}
}

\vspace{-0.3cm}
\subsection{The Successive Convex Approximation Algorithm}\label{subAlgorithm}

In this subsection, we  develop two algorithms and analyze their properties.

Our approach is to successively approximate $f(\bV)$ using
$h^{\beta}(\bV;\bhV)$ to obtain progressively improved solutions. Let us use
$(t)$ to denote the iteration index. The first algorithm, referred to as the {\it successive convex approximation (SCA)} algorithm, consists of three main steps; also see Fig. \ref{figSCA} for a graphical illustration.

{\bf Step 1}: Suppose that $\bV(t-1)$ is a feasible solution to
\eqref{problemSyS}. At iteration $t$, the following convex
optimization problem is solved to obtain $\bV(t)${\small
\begin{subequations}
\begin{align}
\max_{\bV}&\quad
h^{\beta}(\bV;\bV(t-1))-s(\bV)\label{problemLowerBoundGeneral}\tag{${\rm P}_{\rm Lower-Bound}$} \\
\st&\quad \bV^{q_k}\in\mathcal{V}^{q_k},
\forall~q_k\in\mathcal{Q}, \quad\bV^{k}\in\mathcal{V}^{k}, \forall~k\in\mathcal{K}\nonumber.
\end{align}\end{subequations}}
{\bf Step 2}: For each user $i_k\in\cI$, compute{
\begin{subequations}
\begin{align}
\bC_{i_k}(t)&=\sum_{(\ell,j)}\bH^{\ell}_{i_k}\bV_{j_\ell}(t)\bV^H_{j_\ell}(t)(\bH^{\ell}_{i_k})^H
+\sigma^2\bI_N\label{eqCUpdate}\\
{\bU}_{i_k}(t)&=\bC^{-1}_{i_k}(t)\bH^k_{i_k}\bV_{i_k}(t)\label{eqUUpdate}\\
\bE_{i_k}(t)&=\bI_{d}-\bV^H_{i_k}(t)(\bH^k_{i_k})^H\bC^{-1}_{i_k}(t)\bH^k_{i_k}\bV_{i_k}(t)\\
c_{i_k}(t)&=\frac{\partial f_{i_k}(x)}{\partial
x}\big|_{x=-\log|{\bE}_{i_k}(t)|}\label{eqcUpdate}\\
\bJ^k(t)&=\sum_{j_l\in\mathcal{I}}c_{j_l}(t)(\bH^{k}_{j_l})^H\bU_{j_l}(t)\bE^{-1}_{j_l}(t)\bU^H_{j_l}(t)\bH^{k}_{j_l}+\beta\bI_{MQ_k}\label{eqJUpdate}.
\end{align}\end{subequations}}
{\bf Step 3}: Form the updated function
$h^{\beta}(\bV;\bV(t))$ according to \eqref{eqHExpanding}. Let $t=t+1$,
go to {\bf Step 1}.

{The key step in the SCA is to solve the lower-bound maximization subproblem \eqref{problemLowerBoundGeneral} to global optimality. In practice, solving \eqref{problemLowerBoundGeneral} can be computationally expensive. Therefore it is desirable to inexactly optimize the subproblem \eqref{problemLowerBoundGeneral}. The following alternative, termed as {\it inexact successive convex approximation} (In-SCA) algorithm, replaces Step 1 of the SCA by the following step:

{{\bf Step 1 of In-SCA}: Suppose $\bV(t-1)$ is feasible to
\eqref{problemSyS}. At iteration $t$, approximately solve the subproblem \eqref{problemLowerBoundGeneral} to obtain $\bV(t)$, which satisfies conditions \eqref{eq:condition_inexact}.}}
\begin{table*}[t]
\normalsize{ \vspace{-0.5cm}
{
\begin{subequations}\label{eq:condition_inexact}
\begin{align}
&\left(h^{\beta}(\bV(t);\bV(t-1))-s(\bV(t))\right)-\left(h^{\beta}(\bV(t-1);\bV(t-1))-s(\bV(t-1))\right)\ge \eta(t) \|\bV(t-1)-\bV(t)\|^2_F \label{eq:condition_inexact1}
\\
&\; \mbox{if}\quad \bV(t)=\bV(t-1), \quad \mbox{then}\quad \bV(t-1)=\arg\max_{\bV}\quad
h^{\beta}(\bV;\bV(t-1))-s(\bV).\label{eq:condition_inexact2}
\end{align}
\end{subequations}
}
\rule{\linewidth}{0.2mm} }
\vspace{-0.5cm}
\end{table*}

Here $\eta(t)>\eta>0$ is some iteration dependent constant bounded away from zero. The first condition \eqref{eq:condition_inexact1} requires that the new precoder should {\it sufficiently} improve the lower bound, while the second condition \eqref{eq:condition_inexact2} restricts the way in which the iterates should be generated by the algorithm. More specifically \eqref{eq:condition_inexact2} says that if the iteration stops, then a fixed point of the SCA algorithm must have been reached. Note that the In-SCA algorithm is rather a family of algorithms that generates the precoders satisfying the conditions \eqref{eq:condition_inexact}.

Next we analyze the convergence properties of the two proposed algorithms. To proceed, the following definitions are needed:
\begin{itemize}
\item \textbf{Directional derivative:} Let $f: \mathcal{V} \rightarrow \mathbb{R} $ be a
function where~$\mathcal{V}$ is a convex set. The directional
derivative of~$f$ at point~$x\in\cV$ in direction~$v$ is defined by{
\[
f_v'(x) := \liminf_{r \downarrow 0} \frac{f(x+r v) - f(x)}{
r}.
\]}
\item \textbf{Stationary points of a function:} Let $f: \mathcal{V} \rightarrow \mathbb{R} $
where~$\mathcal{V}$ is a convex set.  The point~$x$ is a stationary
point of~$f(\bullet)$ if $f_v'(x) \leq 0$ for all $v$ such that $x + v
\in \cV$.

%
\end{itemize}

      \begin{figure*}[ht]
    \begin{minipage}[t]{0.4\linewidth}\vspace*{-0.3cm}
    \centering
     {\includegraphics[width=
1\linewidth]{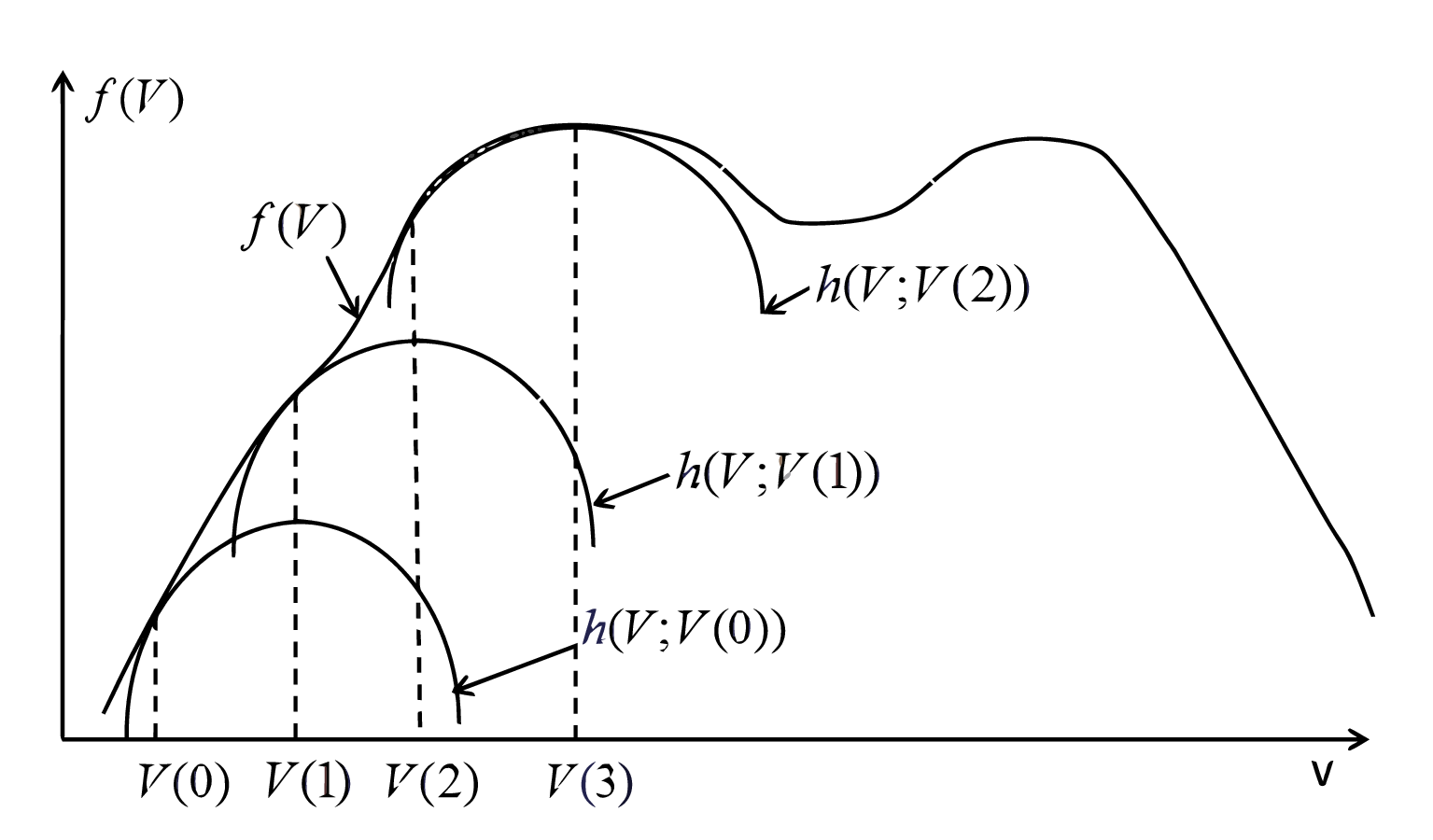}\vspace*{-0.1cm}}
\end{minipage}\hfill
    \begin{minipage}[t]{0.4\linewidth}\vspace*{-0.3cm}
    \centering
    {\includegraphics[width=
1\linewidth]{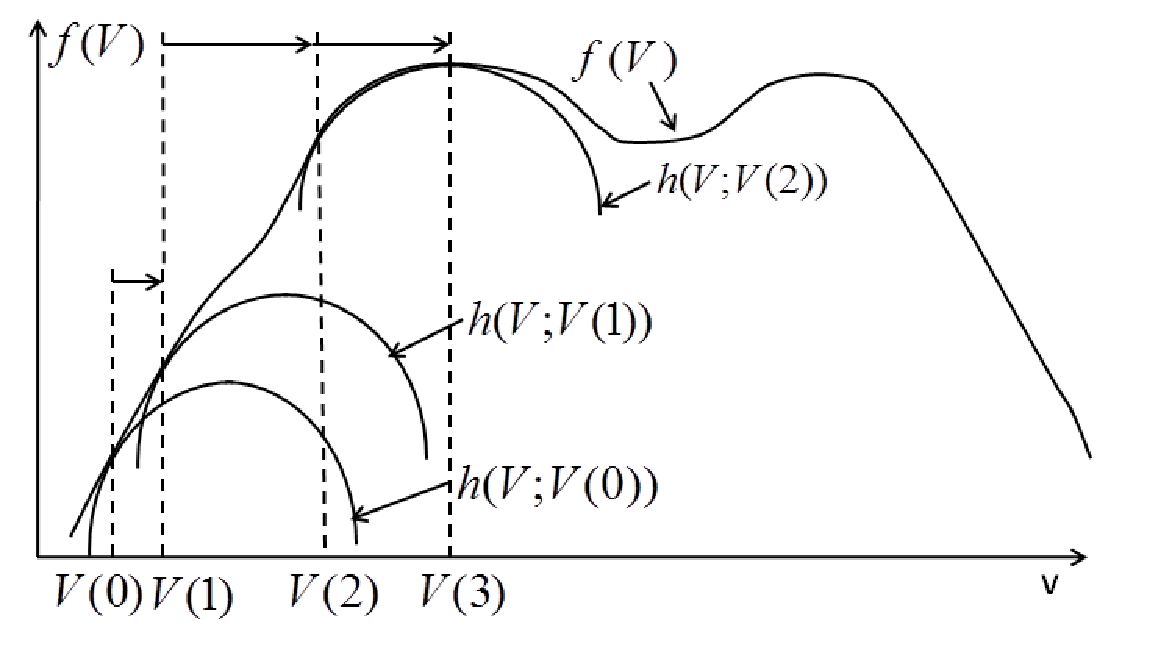} }
\end{minipage}
\caption{\footnotesize A graphical illustration of the proposed algorithms. Left: The SCA Algorithm. Right: The In-SCA Algorithm.}\label{figSCA}
\vspace{-0.5cm}
    \end{figure*}

We have the following convergence theorem. The proof is given in Appendix \ref{appTheorem1}.
\begin{thm}\label{thmGeneralConvergence}
\it Suppose the following conditions hold:
\begin{enumerate}
\item Assumptions A-1)-- A-4) are all satisfied;
\item The sets
$\{\mathcal{V}^k\}_{k\in\cK}$ and
$\{\mathcal{V}^{q_k}\}_{q_k\in\mathcal{Q}}$ are all convex, closed
and compact.
\end{enumerate}
{Then the precoders $\{\bV(t)\}$ generated by either the SCA or the In-SCA globally converge to the set of stationary solutions of problem \eqref{problemSyS}.}
\end{thm}

{
It is useful to compare the SCA and In-SCA with various existing approaches for parallel optimization. Essentially, SCA falls in the BSUM framework \cite{Razaviyayn12SUM,hong15busmm_spm}, but with a {\it single} block variable $\bV$. To see this, recall that ${h}^{\beta}(\bV; \bhV)-s(\bV)$ is a globally tight lower bound of the utility function $u(\bV)$ (cf. \eqref{eqLowerBoundPropertyInequality}), therefore it satisfies the condition A1-A4 in \cite{Razaviyayn12SUM}. In this sense, the SCA can also be viewed as closely related to the {FLEXA (Flexible Parallel Algorithm)} and the {PSCA (Parallel Successive Convex Approximation)} algorithms in \cite{meisam14nips, Facchinei15}, again with a {\it single} block variable. The difference with PSCA and FLEXA is that the decomposability of our algorithms is an intrinsic feature of the approximation function ${h}^{\beta}(\bV; \bhV)-s(\bV)$, not a result of algorithm design. Therefore our algorithms are stepsize-free, while both FLEXA and PSCA requires careful stepsize tuning.
}

{
Here we note that the In-SCA algorithm possesses the same convergence guarantee (to the set of stationary solutions) as the SCA algorithm, despite the fact that its subproblems are solved inexactly. Also note that in Step 1 of both algorithms there are $K$ independent subproblems, one for each cell $k$. This is a result of the decomposability of the approximation function $h^{\beta}(\bullet)$. Clearly, efficiently solving these per-cell problems is the key to the low-complexity implementation of our algorithm. In the next section we shall develop customized algorithms  tailored for the per-cell problems of different applications.
}

\section{Customized Algorithms for Precoder
Design}\label{secCustomizedAlgorithm}
In this section, we customize the SCA and the In-SCA algorithms to different network settings.
\subsection{Linear Precoder Design for IBC Model}\label{subIBCLinear}
First let us consider the IBC model in which there is a single BS in each cell.  As each cell consists of a single BS, $\bV_{i_k}$ is used to denote the transmit precoder intended for user $i_k$. In this case, $s(\bV)\equiv 0$, and the sets $\cV^{q_k}$ and
$\cV^k$ collapse to a single set, cf. \eqref{eqSumPower}. The subproblem \eqref{problemLowerBoundGeneral} is
given by{
\begin{align}
\max_{\{\bV^k\}_{k\in\cK}}&\quad
\sum_{k\in\mathcal{K}}\sum_{i_k\in{\mathcal{I}_k}}g^{\beta}_{i_k}(\bV_{i_k};
\bhV)\label{problemIBC}\tag{${\rm P}_{\rm IBC}$}, \nonumber\\
\st&\quad \sum_{i_k\in\mathcal{I}_k}\trace(\bV_{i_k}\bV^H_{i_k})\le
\bar{P}_k,~k\in\mathcal{K}\nonumber.
\end{align}}
\!\!Clearly both the constraints and the objective of the above problem
are separable among the BSs. Therefore we can decompose this problem into $K$
independent subproblems of the form{
\begin{align}
\max_{\bV^k}&\quad \sum_{i_k\in{\mathcal{I}_k}}g^{\beta}_{i_k}(\bV_{i_k},\bhV)\label{problemIBC-SUB}\tag{${\rm P}_{\rm IBC-SUB}$},\nonumber\\
\st&\quad \sum_{i_k\in\mathcal{I}_k}\trace\left[\bV_{i_k}\bV^H_{i_k}\right]\le
\bar{P}_k\nonumber.
\end{align}}
\!\!Let $\lambda^k\ge 0$ denote the Lagrangian multiplier associated
with the power constraint. Then the Lagrangian function for problem
\eqref{problemIBC-SUB} is given by \eqref{eq:Lagrangian}.

\begin{table*}[t]
{
\begin{align}
&L_k(\bV^k,\lambda^k;\bhV)=\sum_{i_k\in\mathcal{I}_k}
\bigg(\trace\bigg[2(\hc_{i_k}\widehat{\bE}_{i_k}^{-1}\widehat{\bU}^H_{i_k}
\bH^k_{i_k}+\beta\bhV^H_{i_k})\bV_{i_k}-\bV_{i_k}^H\bhJ^k\bV_{i_k}\bigg]\bigg)-\lambda^k\bigg(\sum_{i_k\in\cI_k}\trace\left[\bV_{i_k}\bV^H_{i_k}\right]-\bar{P}_k\bigg).\label{eq:Lagrangian}
\end{align}
\rule{\linewidth}{0.2mm} }
\vspace{-0.8cm}
\end{table*}

{
Clearly the Slater Condition holds, so the optimal primal-dual pair $((\mathbf{V}^k)^*,(\lambda^k)^*)$
satisfy the KKT optimality conditions{
\begin{subequations}\begin{align}
&\bV^{k^*}=\arg\max_{\bV^k} L_k(\bV^k,\lambda^{k*};\bhV),\label{eqLagrangian0}\quad \lambda^{k*}\ge 0\nonumber\\ &\bar{P}_k-\sum_{i_k\in\cI_k}\trace[\bV^*_{i_k}\bV^{*H}_{i_k}]\ge
0,\\
&\lambda^{k*}\left(\bar{P}_k-\sum_{i_k\in\cI_k}\trace[\bV^*_{i_k}\bV^{*H}_{i_k})]\right)=0\label{eqComplementarity}.
\end{align}\end{subequations}}}
\!\!For fixed $\lambda^k\ge 0$, the solution for the unconstrained
problem $\max_{\mathbf{V}^k}L(\bV^k,\lambda^k;\bhV)$, denoted as
$\{\bV_{i_k}^*(\lambda^k)\}_{i_k\in\cI_k}$, can be expressed as{\small
\begin{align}
\bV_{i_k}^*(\lambda^k)&=\left(\bhJ^k+\lambda^k\bI_{M}\right)^{-1}\left(\hc_{i_k}(\bH^k_{i_k})^H
\widehat{\bU}_{i_k}\widehat{\bE}_{i_k}^{-1}+\beta \bhV_{i_k}\right)
,\; \forall~i_k. \label{eqIBCComputeV}
\end{align}}
\!\!To find the optimal multiplier $(\lambda^k)^*$ that satisfies the
complementarity condition \eqref{eqComplementarity}, we utilize a
result on penalty method for optimization,  e.g., \cite[Section 12.1, Lemma 1]{Luenberger984}.  This result asserts
that for the solution $\{\bV_{i_k}^*(\lambda^k)\}_{i_k\in\cI_k}$,
the penalized term
$\sum_{i_k\in\cI_k}\trace[\bV^*_{i_k}(\lambda^k)(\bV^*_{i_k}(\lambda^k))^{H}]$
must be monotonically decreasing with respect to $\lambda^k$. It follows that we can find the optimal multiplier
by a simple bisection search procedure; {see e.g., \cite[Section 3.3.1]{hong12survey} for details on such procedure.}

The algorithm discussed in this subsection is summarized in Table
\ref{tableIBC}. {We can verify that the overall per-iteration computational complexity is given by (assuming $d=N$)
$$\mathcal{O}(I^2 NM^2+ I M^3+ IM^2N + I N^2M +I N^3).$$
}
\begin{table}[htb]
\begin{center}
\vspace{-0.1cm} \caption{ SCA for
\eqref{problemSyS} in IBC setting} \label{tableIBC} {
\begin{tabular}{|l|}
\hline
S1): {\bf Initialization} Obtain a feasible solution $\bV_{i_k}(0)$ for all $i_k$\\
S2): For each $k$, compute $\bC_{i_k}(t),\bE_{i_k}(t),\bU_{i_k}(t),c_{i_k}(t)$,\\
\quad\quad according to \eqref{eqCUpdate}--\eqref{eqcUpdate}, for
all $i_k\in\cI_k$\\
S3): For each $k$, compute $\bJ^k(t)$ by \eqref{eqJUpdate};\\
\quad~~ For each $k$, and $i_k\in\cI_k$, compute $\bV^k(t)$ by\\
\quad\quad
{$\bV_{i_k}(t)=\left(\bJ^k(t)+(\lambda^k)^{*}\bI_M\right)^{-1}$}\\
\quad\quad\quad\quad\quad{ $\times\left(c_{i_k}(t)(\bH^k_{i_k})^H
{\bU}_{i_k}(t){\bE}^{-1}_{i_k}(t)+\beta\bV_{i_k}(t-1)\right)$}.\\
\quad \quad where $(\lambda^k)^*$ is computed by a bisection procedure\\
S4) Until some stopping criterion is met\\
  \hline
\end{tabular}}
\vspace{-0.5cm}
\end{center}
\end{table}


\begin{remark}\label{remarkWMMSE}
{\it When we set $\beta=0$, the algorithm listed in Table \ref{tableIBC} recovers the WMMSE algorithm proposed in \cite{shi11WMMSE_TSP}. It is interesting to understand the difference between the WMMSE approach and the SCA approach studied here. Reference \cite{shi11WMMSE_TSP} starts with a SUM problem with the objective  $u(\bV)$. Then it uses a ``dimension lifting" approach that adds two extra variables $\bW$ and $\bU$, and considers a different problem $\min\; \bar{u}(\bU,\bW,\bV)$. The authors show that these two problems are equivalent, so they apply a block coordinate descent (BCD) algorithm to solve the reformulated three-block problem. However, it is by no means clear what is the intuition behind {\it adding} these two variables. Our SCA framework well explains this question: $\bW$ is in fact the coefficients derived from the first-stage  {linear approximation}, while $\bU$ is the coefficients derived from the second-stage {convex approximation}. Note that in this paper we arrive at the WMMSE by specializing the SCA algorithm to the IBC setting. Thus the SCA algorithm is more general and covers the WMMSE as a special case.} \hfill $\square$
\end{remark}

%
\vspace{-0.3cm}
\subsection{Intra-cell ZF and Inter-cell CB
for IBC Model}\label{subIBCZF}

Consider an IBC model in which each BS employs a ZF precoder to cancel the intra-cell
interference. We assume that certain user selection
within each cell has already taken place, so that the following zero-forcing constraint is always feasible
{\begin{align}
\bH^k_{j_k}\bV_{i_k}\bV_{i_k}^H(\bH^k_{j_k})^H={\bf
0},\forall~j_k\ne i_k, \ j_k\in\mathcal{I}_k.
\end{align}}
\!\!Note that in this network setting, the inter-cell interference is still present, despite the fact that the intra-cell interference is canceled by the use of ZF precoder. Therefore the original problem \eqref{problemSyS} is
still difficult to solve. To apply the SCA algorithm, we first specialize the subproblem
\eqref{problemLowerBoundGeneral} to the following {\begin{subequations}
\begin{align}
\max_{\bV}&\quad
\sum_{k\in\mathcal{K}}\sum_{i_k\in{\mathcal{I}_k}}g^{\beta}_{i_k}(\bV_{i_k};
\bhV)\label{problemIBC-ZF1}\tag{${\rm P_{\rm IBC-ZF1}}$} \\
\st&\quad
\sum_{i_k\in\mathcal{I}_k}\trace\left[\bV_{i_k}\bV_{i_k}^H\right]\le
\bar{P}_k,~k\in\mathcal{K}\\
&\quad\bH^k_{j_k}\bV_{i_k}\bV_{i_k}^H(\bH^k_{j_k})^H={\bf
0},\forall~j_k\ne i_k, \ j_k\in\mathcal{I}_k \label{eqZFConstraint}.
\end{align}\end{subequations}}
\!\!To remove the ZF constraints \eqref{eqZFConstraint}, let
$L:= N(|\cI_k|-1)$, and define new channel
matrices as:{
\begin{align}
\bQ_{i_k}:=\{\bH^k_{j_k}\}_{j_k\in\cI_k\setminus\{i_k\}}\in\mathbb{C}^{L\times
M}, \ \forall~i_k\in\cI_k.\label{eqDefQ}
\end{align}}
\!\!Let $\bQ_{i_k}=\bL_{i_k}\bfSigma_{i_k}\bR^H_{i_k}$ denote the singular
value decomposition of $\bQ_{i_k}$, where $\bL_{i_k} \in \mathbb{C}^{L\times L}$ and $\bR_{i_k} \in \mathbb{C}^{M\times L}$ are two unitary matrices, and $\bfSigma_{i_k} \in \mathbb{R}^{L\times L}$ is a nonnegative diagonal matrix. Define
$\bP_{i_k}:=(\bI-\bR_{i_k}\bR^H_{i_k})$ as a projection matrix to
the space orthogonal to the one spanned by $\bR_{i_k}$. Let
$\bP_{i_k}=\btR_{i_k}\btR^H_{i_k}$, where
$\btR_{i_k}\in\mathbb{C}^{M\times (M-L)}$ is composed of the
orthogonal basis that satisfies $\bR^H_{i_k}\btR_{i_k}={\bf 0}$ and
$\btR^H_{i_k}\btR_{i_k}=\bI$. Then \cite[Lemma 3.1]{zhang10JSAC}
asserts that the optimal solution of problem \eqref{problemIBC-ZF1} must be
of the form: $\bV_{i_k}=\btR_{i_k}\bW_{i_k}$, with
$\bW_{i_k}\in\mathbb{C}^{(M-L)\times d}$. The structure of the optimal solution implies problem
\eqref{problemIBC-ZF1} can be equivalently written as{
\begin{subequations}
\begin{align}
\max_{\bW}&\quad
\sum_{k\in\mathcal{K}}\sum_{i_k\in{\mathcal{I}_k}}\tg^{\beta}_{i_k}(\bW_{i_k};
\bhV)\label{problemIBC-ZF2}\tag{$\rm P_{\rm
IBC-ZF2}$}\nonumber\\
\st&\quad
\sum_{i_k\in\mathcal{I}_k}\trace\left[\btR_{i_k}\bW_{i_k}(\btR_{i_k}\bW_{i_k})^H\right]\le
\bar{P}_k,~k\in\mathcal{K}\nonumber,
\end{align}\end{subequations}}
\!\!where the function $\tg^{\beta}_{i_k}(\bullet)$ is the same as the original
objective $g^{\beta}_{i_k}(\bV_{i_k};\bhV)$, except that $\bV_{i_k}$ is
replaced by $\btR_{i_k}\bW_{i_k}$. Once again both the constraints and the objective are
separable among the BSs, so the problem further decomposes into $K$ independent subproblems: {
\begin{subequations}
\begin{align}
\max_{\bW^k}&\quad \sum_{i_k\in{\mathcal{I}_k}}\tg^{\beta}_{i_k}(\bW_{i_k};
\bhV)\label{problemIBC-ZF-SUB}\tag{${\rm P}_{
\rm IBC-ZF-SUB}$}\nonumber\\
\st&\quad
\sum_{i_k\in\mathcal{I}_k}\trace\left[\btR_{i_k}\bW_{i_k}(\btR_{i_k}\bW_{i_k})^H\right]\le
\bar{P}_k\nonumber.
\end{align}\end{subequations}}
\!\!Let us use $\lambda_k$ to denote the Lagrangian multiplier
associated with the power constraint. Then by using similar steps
leading to \eqref{eqIBCComputeV}, we can show that the optimal
solution for problem \eqref{problemIBC-ZF-SUB} is of the form{\small
\begin{align}
\bW_{i_k}^*&=\left(\sum_{(\ell,j)}\hc_{j_\ell}(\bH^{k}_{j_\ell}\btR_{i_k})^H\bhU_{\ell_j}\bhE_{j_\ell}^{-1}\bhU^H_{\ell_j}\bH^{k}_{j_\ell}\btR_{i_k}+\lambda^*_k\bI_{M-L}\right)^{-1}\nonumber\\
&\quad \times \btR^H_{i_k}\left(\hc_{i_k}(\bH^k_{i_k})^H
\widehat{\bU}_{i_k}\widehat{\bE}_{i_k}^{-1}+\beta\bhV_{i_k}\right),\label{eqIBCComputeW}
\end{align}}
\!\!where $\lambda^*_k$ can be computed by the bisection method.
The algorithm is summarized in Table \ref{tableIBC-ZF}.

\begin{table}[htb]
\begin{center}
\vspace{-0.1cm} \caption{ SCA for
\eqref{problemSyS} in IBC setting with intra-cell ZF and inter-cell
CB } \label{tableIBC-ZF} {
\begin{tabular}{|l|}
\hline
S1): {\bf Initialization}\\
\quad\quad For each $k$ and $i_k$, compute:  \\
\quad\quad \quad $\bQ_{i_k}=\bL_{i_k}\bfSigma_{i_k}\bR^H_{i_k}$ \\
\quad\quad \quad $\bP_{i_k}=(\bI-\bR_{i_k}\bR^H_{i_k})$}, {$\bP_{i_k}=\btR_{i_k}\btR^H_{i_k}$; \\
\quad\quad Obtain a feasible solution $\bV_{i_k}(0)$ for all $i_k$;\\

S2): For each BS $k$, compute $\bC_{i_k}(t),\bE_{i_k}(t),\bU_{i_k}(t),c_{i_k}(t)$,\\
\quad\quad according to \eqref{eqCUpdate}--\eqref{eqcUpdate}, for
all $i_k$; \\

S3): Compute $\bW_{i_k}(t)$ according to \eqref{eqIBCComputeW}.\\
\quad \quad Let { $\bV_{i_k}(t)=\btR_{i_k}\bW_{i_k}(t)$};\\
S4) Until
some stopping criterion is met.\\
  \hline
\end{tabular}}
\vspace{-0.5cm}
\end{center}
\end{table}

\vspace{-0.2cm}
\subsection{Linear Precoder Design for HetNet with Intra-Cell Full CoMP and Inter-Cell
CB}\label{subVIBC} Consider a HetNet setting in which there are a set of $\cQ_k$ BSs in each cell, and they form a single {\it
virtual} BS to transmit to the users. In this case, $\cV^{q_k}$ is given by \eqref{eqPerBSPower}, i.e., each BS $q_k$ has its own power constraint. {This scenario also covers the multi-cell IBC scenario with per group of antenna power constraints, see e.g., \cite{yu07perantenna}}.

Assume for now that there is no penalty term $s(\bV)$.  Then the subproblem \eqref{problemLowerBoundGeneral}
again decomposes into $K$ independent subproblems:{
\begin{align}
\max_{\bV^k}&\quad\sum_{i_k\in{\mathcal{I}_k}}g^{\beta}_{i_k}(\bV_{i_k},\bhV)\label{problemVIBC-SUB}\tag{${\rm P}_{\rm VIBC-SUB}$}\nonumber\\
\st&\quad
\sum_{i_k\in\mathcal{I}_k}\trace[\bV^{q_k}_{i_k}(\bV^{q_k}_{i_k})^H]\le
\bar{P}^{q_k},\ \forall~q_k\in \cQ_k. \nonumber
\end{align}}
\!\!Differently from problem \eqref{problemIBC-SUB}, the above problem has $Q_k=|\cQ_k|$ separable
constraints (each constraining a subset of variables), hence $Q_k$
Lagrangian multipliers $\{\lambda^{q_k}_k\}_{q_k\in\cQ_k}$. Therefore the
bisection algorithm on a single multiplier no longer works.

Fortunately, the constraints for this problem are separable among
different block variables $\{\bV^{q_k}\}_{q_k\in\cQ_k}$. This leads to a block coordinate descent (BCD) algorithm (see \cite{tseng09, tseng01}), in which one block variable $\bV^{q_k}$ is updated at a time while holding the remaining block variables
fixed. To capitalize the block structure of the problem, the following definitions are helpful.
Let{\small
\begin{align}
\hspace{-0.2cm}\bhS_{i_k}&:=
\hc_{i_k}(\bH^k_{i_k})^H\bhU_{i_k}\bhE_{i_k}^{-1}+\beta\bhV_{i_k}\in\mathbb{C}^{MQ_k\times
d_k},\ \forall~i_k\in\mathcal{I}_k.\label{eqDefS}
\end{align}}\hspace{-0.2cm}
Partition $\bhJ^k$ (cf. \eqref{eqDefJ}) and
$\bhS_{i_k}$ into the following form{\small
\begin{align}
\bhJ^k&=\left[ \begin{array}{lll}
\bhJ^k[1,1],&\cdots,&\bhJ^k[1,Q_k]\\
\vdots&\ddots&\vdots\\
\bhJ^k[Q_k,1]&\cdots&\bhJ^k[Q_k,Q_k]
\end{array}\right],\nonumber\\
\bhS_{i_k}&=\left[\bhS_{i_k}^H[1],\cdots,\bhS_{i_k}^H[Q_k]\right]^H\label{eqPartition}
\end{align}}\hspace{-0.2cm}
where $\bhJ^k[q,p]\in\mathbb{C}^{M\times M},\
\forall~(q,p)\in\mathcal{Q}_k\times\mathcal{Q}_k$, and
$\bhS_{i_k}[q]\in\mathbb{C}^{M\times d},
\forall~q\in\mathcal{Q}_k$. Then the function $g^{\beta}_{i_k}(\bV_{i_k};
\bhV)$ defined in \eqref{eqHExpanding} can be alternatively expressed as{
\begin{align}
g^{\beta}_{i_k}(\bV_{i_k};
\bhV)&=\widetilde{a}_{i_k}+\sum_{p_k\in\cQ_k}2\trace\left[\bhS^H_{i_k}[p_k]\bV^{p_k}_{i_k}\right]\nonumber\\
&\quad-\sum_{p_k, q_k\in
\cQ_k}\trace\left[(\bV^{q_k}_{i_k})^H\bhJ^k[q_k,p_k]\bV^{p_k}_{i_k}\right]\label{eqGikHetNet}.
\end{align}}
\!\!Now it becomes clear that the objective function of problem \eqref{problemVIBC-SUB}, $\sum_{i_k\in\cI_k}g^{\beta}_{i_k}(\bV_{i_k}; \bhV)$, is a quadratic function with respect to $\bV^{m_k}$, the precoder used by BS $m_k\in\cQ_k$. It follows that the per-block problem, written in the
following form, can be efficiently solved in closed form for each block $m_k\in \cQ_k${
\begin{align}
\max_{\bV^{m_k}}&\sum_{i_k\in{\mathcal{I}_k}}g^{\beta}_{i_k}(\bV_{i_k};\bhV)\label{problemVIBC-BLK}\tag{$\rm P_{\rm VIBC-BLK}$}\nonumber\\
\st&\quad
\sum_{i_k\in\mathcal{I}_k}\trace[\bV^{m_k}_{i_k}(\bV^{m_k}_{i_k})^H]\le
\bar{P}^{m_k}.\nonumber
\end{align}}
Let $\lambda^{m_k}\ge 0$ denote the Lagrangian multiplier associated
with the power constraint of the $m_k$-th subproblem. Following the
same derivation in Section \ref{subIBCLinear}, the optimal solution
$\bV^{m_k*}$ for problem \eqref{problemVIBC-BLK} can be
expressed as{\small
\begin{align}
\bV_{i_k}^{m_k *} &=\left(\bhJ^k[m_k,m_k]+
\lambda^{m_k*}\bI_M\right)^{-1}\nonumber\\
&\times\bigg(\bhS_{i_k}[m_k]-\sum_{p_k\ne
m_k}\bhJ^k[m_k,p_k]\bV^{p_k}_{i_k}\bigg), \
\forall~i_k\in\cI_k\label{eqVIBCComputeV}
\end{align}}
\!\!where the optimal multiplier can be computed again using a bisection search.

The above observation leads to a natural two-layer algorithm: {\it i)} the outer layer updates $\bE_{i_k}(t)$,
$\bU_{i_k}(t)$, $c_{i_k}(t)$, $\bJ^k(t)$ and $\bS^k(t)$; {\it ii)}
the inner layer updates each $\bV^k$ by a BCD algorithm with
blocks given by $\{\bV^{q_k}\}_{q_k\in\cQ_k}$. See Table \ref{tableVIBC} for the detailed description. {Note that in the table $\beta>0$ is some proximal parameter introduced to regularize the iterates}. Once again the convergence of the algorithm is a direct consequence of Theorem \ref{thmGeneralConvergence}.

\begin{table}[htb]
\begin{center}
\vspace{-0.3cm} \caption{SCA for Solving
\eqref{problemSyS} with intra-cell ComP and inter-cell CB}
\label{tableVIBC} {
\begin{tabular}{|l|}
\hline
S1): {\bf Initialization} Obtain a feasible solution $\bV_{i_k}(0)$, $\forall~i_k$\\
S2): For each BS $k$, compute $\bC_{i_k}(t),\bE_{i_k}(t),\bU_{i_k}(t),c_{i_k}(t)$,\\
\quad\quad according to \eqref{eqCUpdate}--\eqref{eqcUpdate}, for
all $i_k\in\cI_k$\\
S3): For each BS $k$ and $i_k\in\mathcal{I}_k$, \\
compute $\bJ^k(t)$ by \eqref{eqJUpdate}; \\
\quad\quad $\bS_{i_k}(t)=
c_{i_k}(t)(\bH^k_{i_k})^H\bU_{i_k}(t)\bE^{-1}_{i_k}(t)+\beta\bV_{i_k}(t-1).$ \\
S4): For each BS $k$, compute the precoders $\bV^k(t)$ by\\
\quad\quad {\bf Repeat} Cyclically pick $m_k\in\cQ_k$\\
 \quad\quad\quad
Compute $\bV^{m_k*}_{i_k}$ using \eqref{eqVIBCComputeV},
$\forall~i_k\in\cI_k$,\\
\quad \quad \quad where $\lambda^{m_k*}$ is computed by a bisection procedure\\
\quad\quad {\bf Until} An optimal solution of \eqref{problemVIBC-SUB}  is obtained.\\
\quad\quad Let $\bV^{m_k}_{i_k}(t)=\bV^{m_k*}_{i_k}$, $\forall~i_k, m_k$\\
 S5) Until
some stopping criteria is met.\\
  \hline
\end{tabular}}
\vspace{-0.5cm}
\end{center}
\end{table}

The key question here is whether in Step S4) one needs to solve the inner problem \eqref{problemVIBC-SUB} {\it exactly} before we can update the outer layer. There are several drawbacks with this approach:
\begin{enumerate}
\item  It is usually difficult to check whether the inner iteration
has indeed reached the optimality;

\item Before reaching the optimality for the subproblem
\eqref{problemVIBC-SUB}, the marginal benefit of the precoder
updates in the inner iteration decreases as the iteration
progresses. This effect is manifested at the first few
outer iterations, in which even the inner problem is solved exactly,
the precoders obtained are still far away from the optimal ones.
\end{enumerate}

Surprisingly, by utilizing the In-SCA algorithm, a {\it single} inner BCD iteration is sufficient to guarantee the convergence of the overall algorithm.   The benefit of such {\it inexact} algorithm is quite obvious from our preceding discussion. It allows one to solve the subproblems approximately at the beginning, and more accurately later as the iteration progresses. To be more specific, all the steps of the resulting algorithm will be the same as before, except for Step S4), where the inexact algorithm only requires each block be picked at least once; see Table \ref{tableVIBC2}. {It can be verified that the overall per-iteration complexity is given by (assuming $d=N$)
$$ \mathcal{O}\left(Q I (M^3+M^2N+QM^2 N) + I^2 NM^2+I N^3+I MN^2\right).$$}

\begin{table}[htb]
\begin{center}
\vspace{-0.3cm} \caption{ In-SCA for Solving
\eqref{problemSyS} with intra-cell CoMP and inter-cell CB}
\label{tableVIBC2} {
\begin{tabular}{|l|}
\hline
S1): {\bf Initialization} Obtain a feasible solution $\bV_{i_k}(0)$ for all $i_k$\\
S2): For each BS $k$, compute $\bC_{i_k}(t),\bE_{i_k}(t),\bU_{i_k}(t),c_{i_k}(t)$,\\
\quad\quad according to \eqref{eqCUpdate}--\eqref{eqcUpdate}, for
all $i_k\in\cI_k$\\
S3): For each BS $k$, compute $\bJ^k(t)$ by \eqref{eqJUpdate}; \\
\quad~~~For each BS $k$, compute \\
\quad\quad $\bS_{i_k}(t)=
c_{i_k}(t)(\bH^k_{i_k})^H\bU_{i_k}(t)\bE^{-1}_{i_k}(t)$\\
\quad\quad\quad\quad\quad$+\beta\bV_{i_k}(t-1), \ \forall~i_k\in\mathcal{I}_k.$ \\S4): For each BS $k$, compute the precoders $\bV^k(t)$ by\\
\quad\quad {\bf Repeat} Cyclically pick $m_k\in\cQ_k$\\
 \quad\quad\quad
Compute $\bV^{m_k*}_{i_k}$ using \eqref{eqVIBCComputeV},
$\forall~i_k\in\cI_k$,\\
\quad \quad \quad where $\lambda^{m_k*}$ is computed by a bisection procedure\\
\quad\quad {\bf Until} Each $m_k$ is picked at least once\\
\quad\quad Let $\bV^{m_k}_{i_k}(t)=\bV^{m_k*}_{i_k}$, $\forall~i_k, m_k$\\
 S5) Until
some stopping criteria is met.\\
  \hline
\end{tabular}}
\vspace{-0.5cm}
\end{center}
\end{table}

{Next we show the convergence of the above inexact algorithm, by appealing to Theorem \ref{thmGeneralConvergence}. The proof is given in Appendix \ref{eq:appC}.
\begin{coro}\label{cor:equivalence}
{\it Consider problem \eqref{problemSyS} specialized to the setting given in this section. Suppose $\beta>0$, then the inexact algorithm in Table \ref{tableVIBC2} converges to the set of stationary solutions. }
\end{coro}
{
\begin{remark}\label{remark:proximal}
The proof of Corollary \ref{cor:equivalence} hinges on the fact that each $\bV^{m_k}$ problem is strongly convex, cf. \eqref{eq:strong:convexity}. This is precisely the reason that we have introduced the proximal term $-\beta\left\|\bV_{i_k}-\bhV_{i_k}\right\|^2_F$ in the first place. As long as $\beta$ is bounded away from zero, the lower bound function $\sum_{i_k\in\cI_k}g^{\beta}_{i_k}(\bV_{i_k}; \bhV)$ is strongly convex with respect to each $\bV^{m_k}$, providing the desired sufficient descent property \eqref{eq:condition_inexact1}. \hfill $\square$%
\end{remark}}

{
\begin{remark}\label{remarkPartialComP}
The algorithms proposed in this
section can be easily extended to the case of per-cell partial CoMP. Assume that the BS clustering structure is known, and we let $\cS^{q_k}\subseteq \cI_k$ denote the set of users served by BS $q_k$. Then we only need to slightly modify the algorithm in Table \ref{tableVIBC} and \ref{tableVIBC2} by the following:
\begin{itemize}
\item  In S1), for each BS $m_k\in\cQ_k$, set
$\bV_{i_k}^{m_k}(0)={\bf 0}$ for all $i_k\notin\cS^{m_k}$;
\item In S4), let each BS $m_k\in\cQ_k$ compute
$\bV_{i_k}^{m_k}(t)$ using \eqref{eqVIBCComputeV},
$\forall~i_k\in\cS^{m_k}$ (rather than using all $i_k\in\cI_k$).
\end{itemize}
In this way, only the precoders of the subset of users served by each BS will be updated at each iteration. Once again, the update at each iteration is closed-form, while in related works such as \cite{Ng10}, general purpose convex solvers are required.

Moreover, when the BSs' clustering structure needs to be designed jointly with the precoders, we can include appropriate penalty term $s(\bV)$
into the objective. Specifically, the per-block subproblem \eqref{problemVIBC-BLK} takes the following form{
\begin{align}
\max_{\bV^{m_k}}&\quad \sum_{i_k\in{\mathcal{I}_k}}g^{\beta}_{i_k}(\bV_{i_k},\bhV)-\sum_{i_k\in\cI_k}s^{m_k}_{i_k}(\bV^{m_k}_{i_k})\nonumber\\
\st&\quad
\sum_{i_k\in\mathcal{I}_k}\trace\left[\bV^{m_k}_{i_k}(\bV^{m_k}_{i_k})^H\right]\le
\bar{P}^{m_k}.\nonumber
\end{align}}
\!\!When we let $s^{m_k}_{i_k}(\bV^{m_k}_{i_k})=\gamma_{i_k}^{m_k}\|\bV^{m_k}_{i_k}\|_F$,
this subproblem becomes a well known quadratic group-LASSO problem
\cite{yuan06} (with an additional quadratic constraint), which can
be solved using a bisection method; see \cite{hong12sparse, hong12_asilomar} for details.
\hfill $\square$
\end{remark}}


\begin{remark}
For the problem discussed in this section, the In-SCA can take other forms as well. For example one can utilize the BSCA proposed in \cite{Razaviyayn12SUM}, or the CGD method proposed in \cite{tseng09coordiate,tseng09}. All that is needed is to verify that the conditions \eqref{eq:condition_inexact} are satisfied. However, it appears that the BCD-based In-SCA proposed in Table \ref{tableVIBC2} takes a much simpler form, and it is much easier to implement and analyze.
\hfill $\square$
\end{remark}

\section{Numerical Results}\label{secSimulation}

In this section we conduct experiments to validate the
effectiveness of the proposed algorithms. We consider three main settings: 1)
Multicell downlink linear precoder design (i.e., the IBC model); 2)
HetNet downlink linear precoder design with inter-cell CB and
intra-cell JP; 3) HetNet downlink joint clustering and linear
precoder design with intra-cell partial CoMP.

The general setup for the experiments are given as follows. {We consider a multi-cell network of up to $10$ hexagonal cells in a square grid}. The distance of the centers of two adjacent cells is set to be $500$ meters (representing a HetNet with densely deployed cells). Both the BSs and the users are
randomly placed in each cell. Let $y^{q_\ell}_{i_k}$ denote the distance between BS $q_\ell$ and user $i_k$. The channel coefficients between user $i_k$ and BS $q_\ell$ are modeled as zero mean circularly symmetric complex Gaussian vector with $\left({200}/{y^{q_\ell}_{i_k}}\right)^{3}L^{q_\ell}_{i_k}$ as variance for both real and imaginary dimensions, where $10\log_{10}(L^{q_\ell}_{i_k})\sim\mathcal{N}(0,64)$ is a real Gaussian random variable modeling the shadowing effect \footnote{The choice of the channel parameters follow those provided in \cite{liao13admm, hong12sparse, venturino10}. Note that commonly accepted standard deviation of shadowing is between $6$dB and $12$dB, and our choice is $8$dB. }. We set the noise power $\sigma^2=1$, and uniformly randomly generate the power budget $\bar{P}^{q_k}\in (0,\;{\rm P}^{\rm tot}_k/|\cQ_{k}|]$ for all $q_k\in \cQ_k$, where the ${\rm P}^{\rm tot}_k>0$ represents an upper bound of the power budget for cell $k$.

The stopping criteria are chosen as follows. The single time-scale
algorithm (i.e. the In-SCA in Table \ref{tableVIBC2}, or the SCA described in Table \ref{tableIBC}, \ref{tableIBC-ZF}) as well as the outer loop of the double time-scale
algorithm (i.e., the SCA described in Table \ref{tableVIBC}) stop when $\frac{|u(t+1)-u(t)|}{|u(t)|}\le 10^{-3}$. The
inner loop of the two-time scale algorithm stops while the relative
increase of the objective value for the related subproblem (i.e.,
problem \eqref{problemVIBC-BLK}) is less than $10^{-3}$ after
performing one round of update by all the BSs in the cell.


\vspace{-0.3cm}
\subsection{HetNet and Multicell Downlink Setting}
In this section, the performance of the following algorithms is
compared:
\begin{enumerate}
\item {\bf WMMSE} \cite{shi11WMMSE_TSP}: The algorithm described in Table \ref{tableIBC} with $\beta=0$.

\item {\bf SCA for HetNet}: The two time-scale
algorithm given in Table \ref{tableVIBC}, with the difference that the inner problem Step S4) is solved using general purpose solvers. 

\item {\bf In-SCA for HetNet}: The inexact
algorithm described in Table \ref{tableVIBC2}, where each block is updated {\it only once} in Step S4).

\item {\bf ZF-SCA for IBC}: The intra-cell ZF plus
inter-cell CB algorithm described in Table \ref{tableIBC-ZF};

\item {\bf Per-Cell ZF for HetNet} \cite{zhang10JSAC}:  Algorithm 2 proposed in \cite{zhang10JSAC}, which
performs intra-cell ZF for the HetNet setting with BS power
constraint, while ignoring the inter-cell interference.
\end{enumerate}

All algorithms considered in this subsection use the sum rate utility. The plots to be shown
represent the averaged performance over
$100$ independent network generations.

Our first set of experiments compare the performance of the first
three algorithms listed above. In Fig. \ref{figRateIBCCellSize}--Fig.
\ref{figTimeIBCCellSized3}, the averaged system sum rate, the averaged CPU time and the total number of iterations used by different algorithms are compared for a network with $|\cQ_k|=6$, $M=5$, $N=3$,
$|\cI_k|=10$ and $P^{\rm tot}_k=20$dB.  For the WMMSE, the per-BS power constraint is completely ignored.
Instead, a {\it single} per-cell power budget is assumed. Several
interesting observations can be made. First, in the HetNet setting the
In-SCA is much more efficient than the SCA. Second, the In-SCA uses almost the same number of iterations as the SCA algorithm, which implies that they also require similar amount of message exchanges among the cells.  Third, our experiments suggest that In-SCA is not quite sensitive to the choice of $\beta$. It even works well when $\beta=0$, although this result is not displayed due to its similarity to the cases when $\beta=\{0.01, 0.05\}$.
\begin{figure}
     {\includegraphics[width=
1\linewidth]{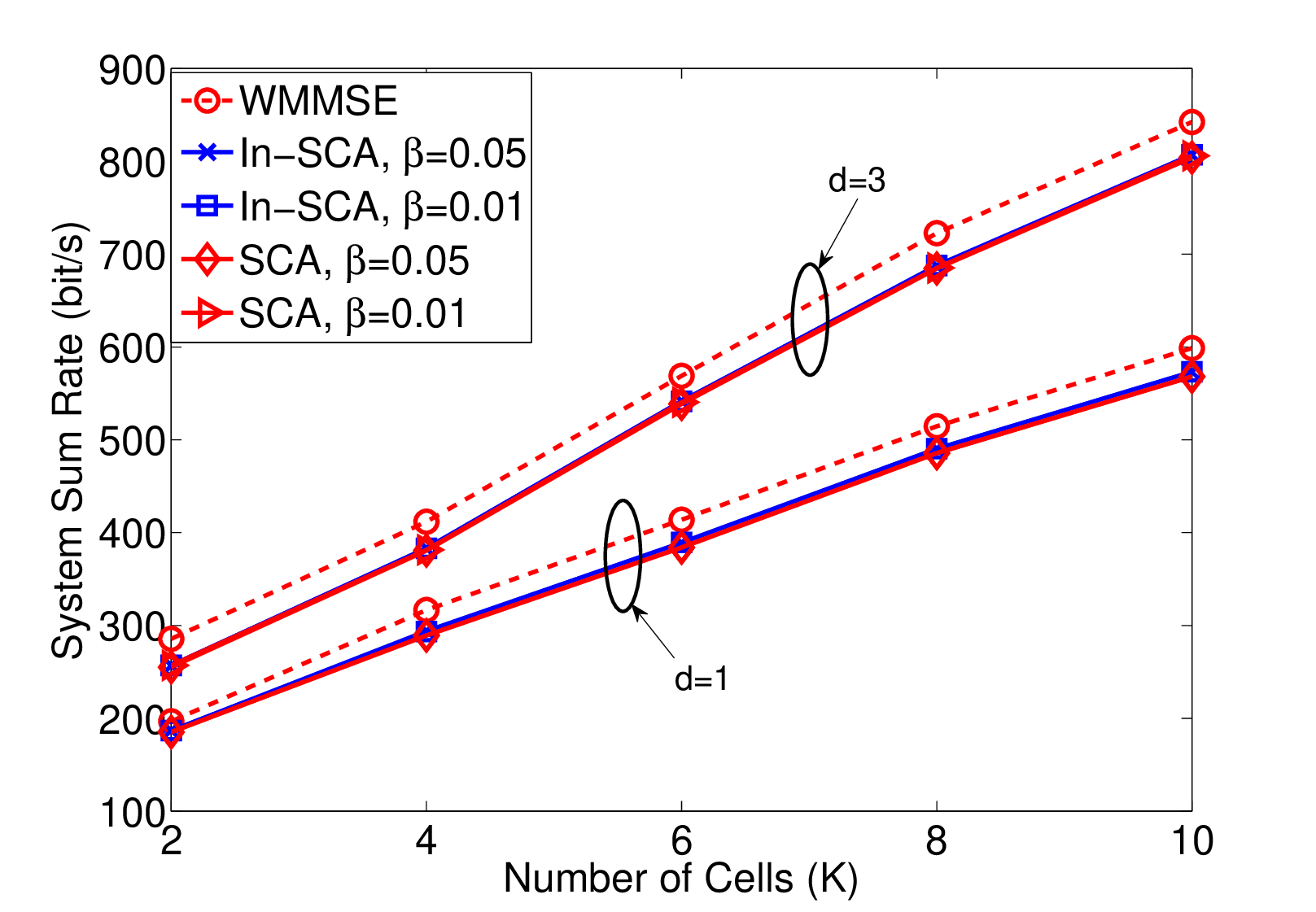}
\vspace{-0.3cm}
\caption{\footnotesize Comparison of the system throughput of
different algorithms in HetNet setting. $K=[2,4,6,8,10]$, $P^{\rm tot}_k=20$dB, $|\cQ_k|=6$,
$|\cI_k|=10$, $M=5$, $N=3$, $d=1$ or $d=3$ for all
$i_k\in\cI$.}\label{figRateIBCCellSize}}\vspace{-0.5cm}
    \end{figure}

      \begin{figure*}[htb]\vspace{-0.6cm}
      \begin{center}
{\subfigure[][Comparing number of iterations]{\resizebox{.48\textwidth}{!}{\includegraphics{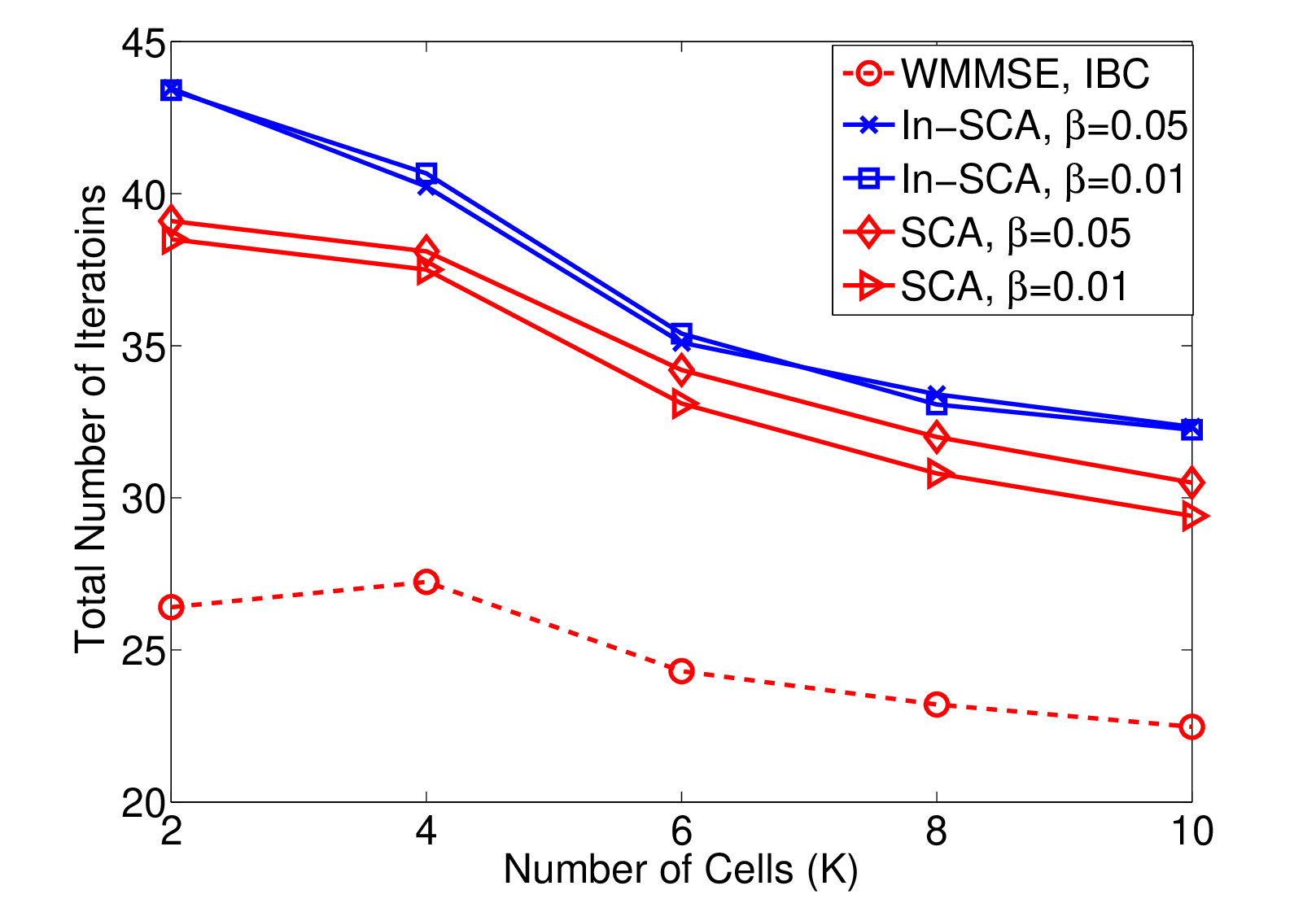}}}}
\hspace{0.3pc}
{\subfigure[][Comparing CPU time]{\resizebox{.48\textwidth}{!}{\includegraphics{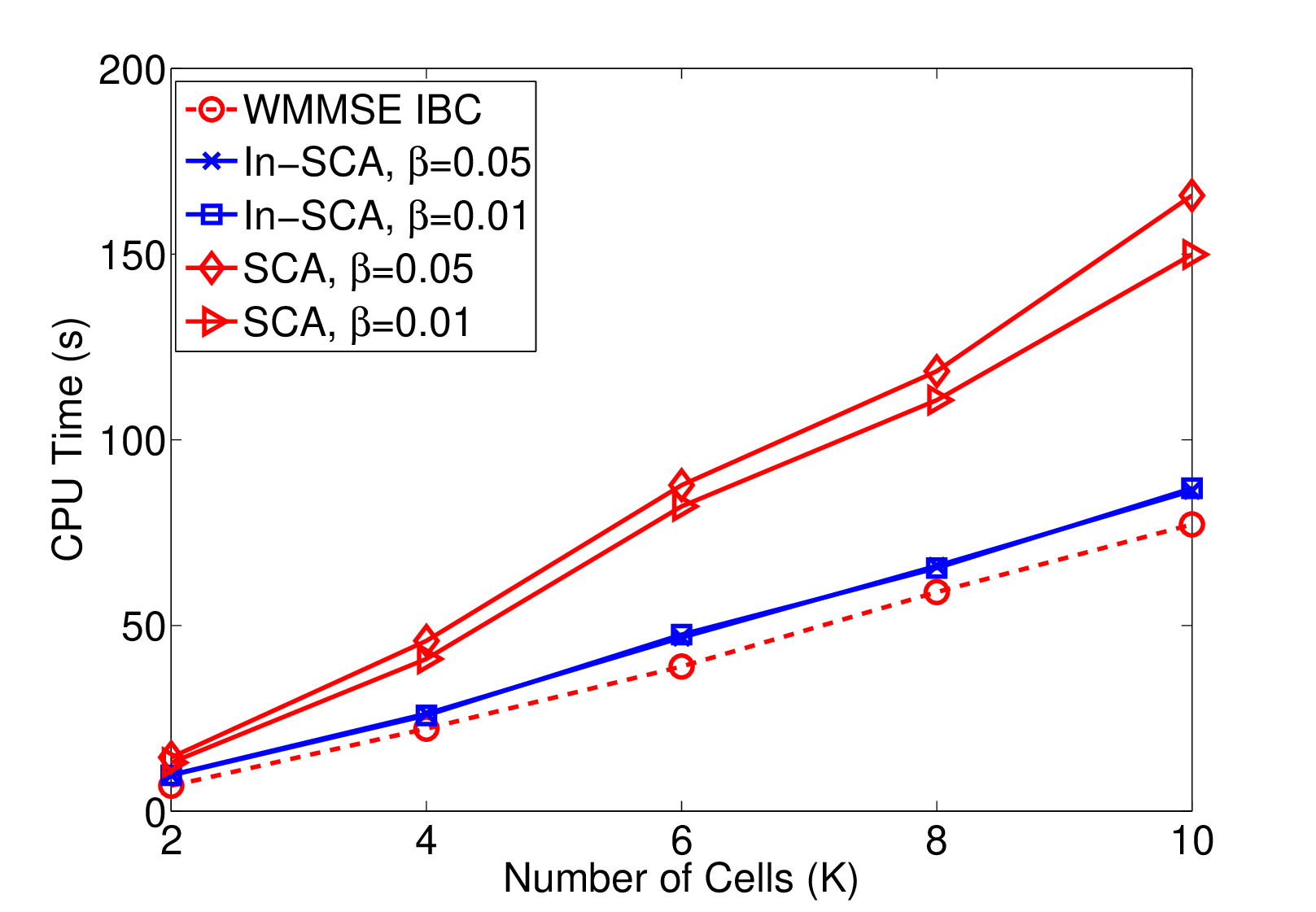}}}}
\end{center}\vspace{-0.1cm}
\caption{ Comparison of the CPU Time and the total iteration number required for
computation in HetNet setting with different sizes of the network.
$K=[2,4,6,8,10]$, $P^{\rm tot}_k=20$dB, $|\cQ_k|=6$, $|\cI_k|=10$,
$M=5$, $N=3$, $d=3$ for all
$i_k\in\cI$.}\label{figTimeIBCCellSized3}
\vspace{-0.5cm}
 \end{figure*}
The second set of experiments compare the general linear precoding and the ZF precoding. In Fig. \ref{figRateChangII}, we show the performance of the algorithms with $K=6$, $|\cQ_k|=6$, $|\cI_k|=\{2, 4, \cdots, 12\}$, $M=\{4, 6\}$, $N=2$ and $d=2$, $P^{\rm tot}_k=20$dB.
Note that when increasing the number of users $|\cI_k|$, more resources are dedicated to eliminating the intra-cell interference. However, as suggested in Fig. \ref{figRateChangII}, this is not an ideal strategy to deal with interference in densely deployed HetNet. The performance of the ZF-based strategies degrades as $|\cI_k|$ increases. Further, when $M=4$ and when $|\cI_k|$ approaches the maximum number of allowable users for which the ZF strategy is still feasible ($12$ in this case), the performance degradation is more severe. The reason is that when the nodes densely deployed, inter-cell interference becomes equally detrimental as the intra-cell interference. The general linear precoding does not pre-specify the source of interference to be mitigated, therefore appears to be a better candidate for dealing with interference.


\begin{figure*}[htb]\vspace{-1cm}
\begin{center}
{\subfigure[][$M=4$]{\resizebox{.48\textwidth}{!}{\includegraphics{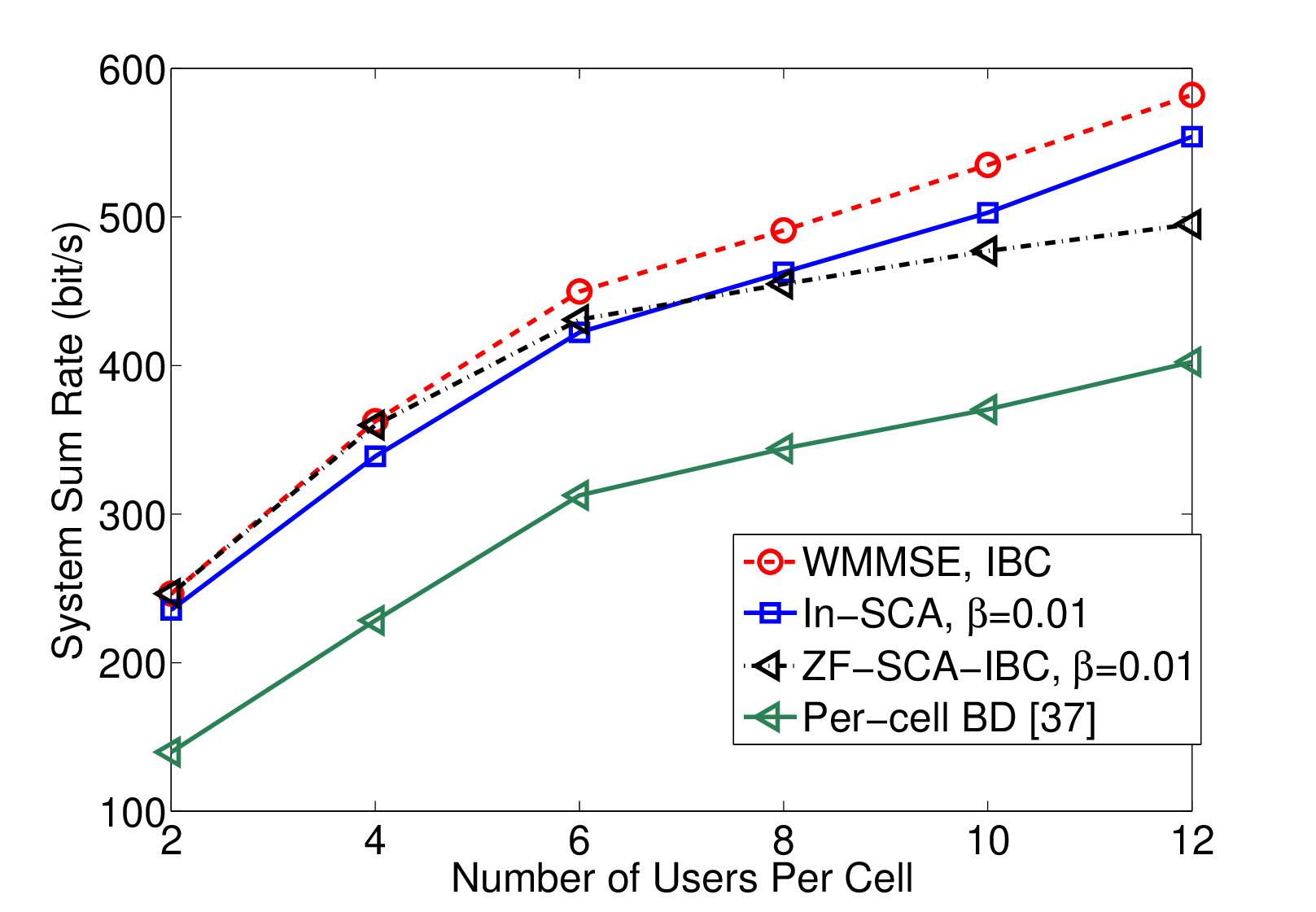}}}}
\hspace{0.3pc}
{\subfigure[][$M=6$]{\resizebox{.48\textwidth}{!}{\includegraphics{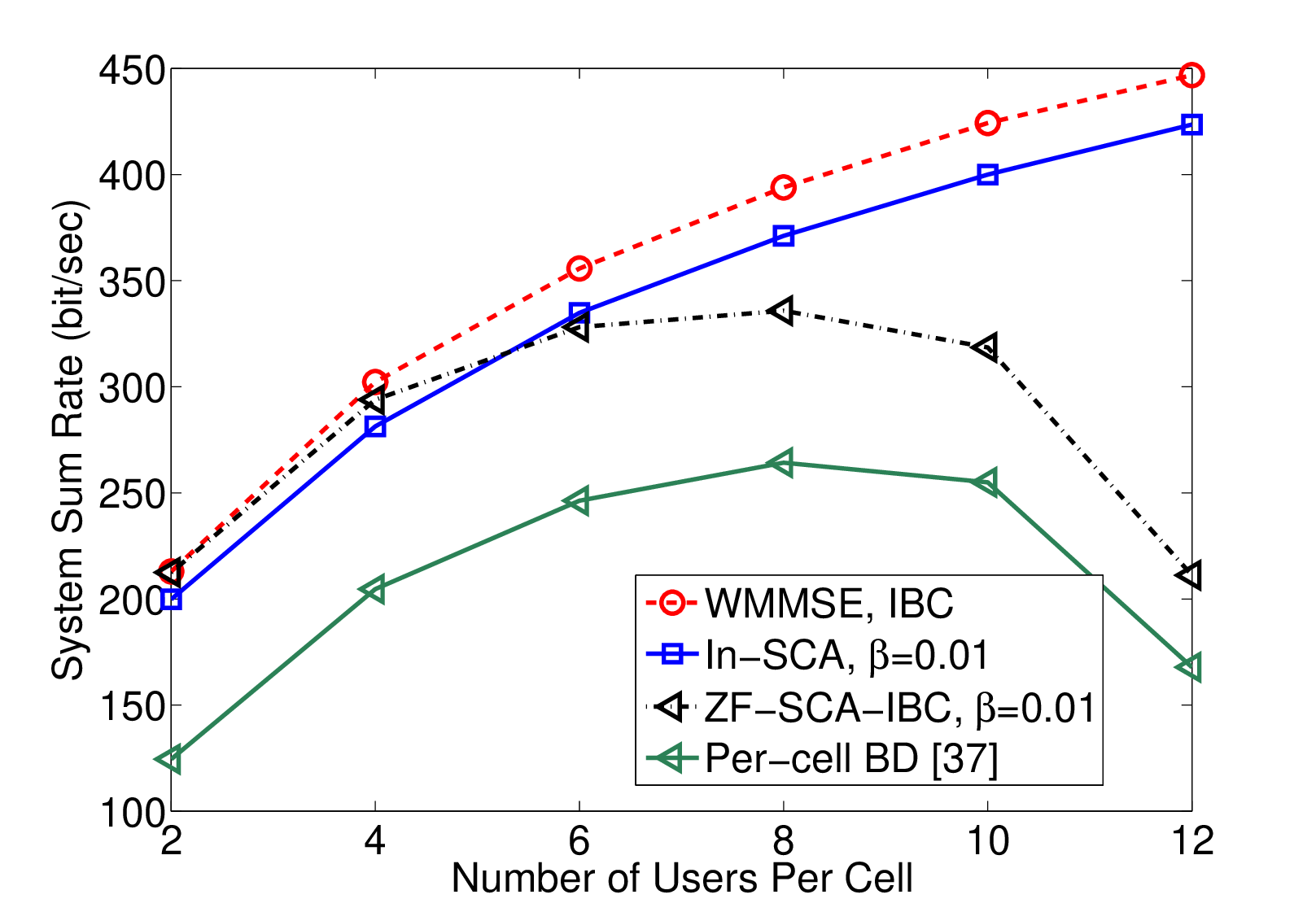}}}}
\end{center}
\vspace{-0.3cm}
\caption{ Comparison of system sum rate achieved for different algorithms with
different number of users per cell. $K=6$, $P^{\rm tot}_k=20$dB,
$|\cQ_k|=6$, $|\cI_k|=[2,4,6,\cdots, 12]$, $N=2$, $d=2$
for all $i_k\in\cI$. Left: performance with $M=6$. Right: performance with $M=4$.}\label{figRateChangII}
\vspace{-0.3cm}
 \end{figure*}

\vspace{-0.3cm}
\subsection{Partial CoMP in HetNet}
In this section, we jointly design the clustering and
linear precoding schemes in a partial CoMP setting. To induce the
desired clustering structure, we specialize the penalty terms in the
objective to take the following form \cite{hong12sparse, hong12_asilomar}:
$s_{i_k}^{q_k}(\bV)=\lambda\|\bV^{q_k}_{i_k}\|_F$, $\forall~i_k,
q_k$ where $\lambda>0$ is chosen appropriately to balance the
resulting group size and the throughput performance. To provide
certain level of fairness among the users, we use the geometric mean utility of one plus the user's
rate, i.e., $u_{i_k}(R_{i_k})=\log(1+R_{i_k})$.

We compare the performance of the following algorithms:
{\bf (1)} {WMMSE}, the baseline algorithm that optimizes the
precoders by treating all the BSs in each cell as a single virtual
BS; {\bf(2)} The algorithm proposed in \cite{hong12sparse} (which can be viewed
as a special case of the SCA algorithm for solving the penalized
utility maximization problem, see Remark \ref{remarkPartialComP}), modified by adding the proximal term with $\beta=0.001$;
{\bf (3)}: In-SCA algorithm in Table \ref{tableVIBC2}, with $\beta=0.001$.

The results are summarized in Fig. \ref{figRateSparse}--\ref{figIterationHetNet} and Table \ref{tableCPUTime}. Both the SCA and the In-SCA based approaches are able to keep a large portion of the system sum rate achieved by the full per-cell
cooperation, while using only small cluster sizes. In contrast, the WMMSE always mandates all the BSs to transmit to each user. The SCA and In-SCA  use similar number of iterations and achieve similar performance. The advantage of using the inexact version is in the efficiency of computation.


\begin{figure*}[htb]\vspace{-0.3cm}
\begin{center}
{\subfigure[][System throughput]{\resizebox{.48\textwidth}{!}{\includegraphics{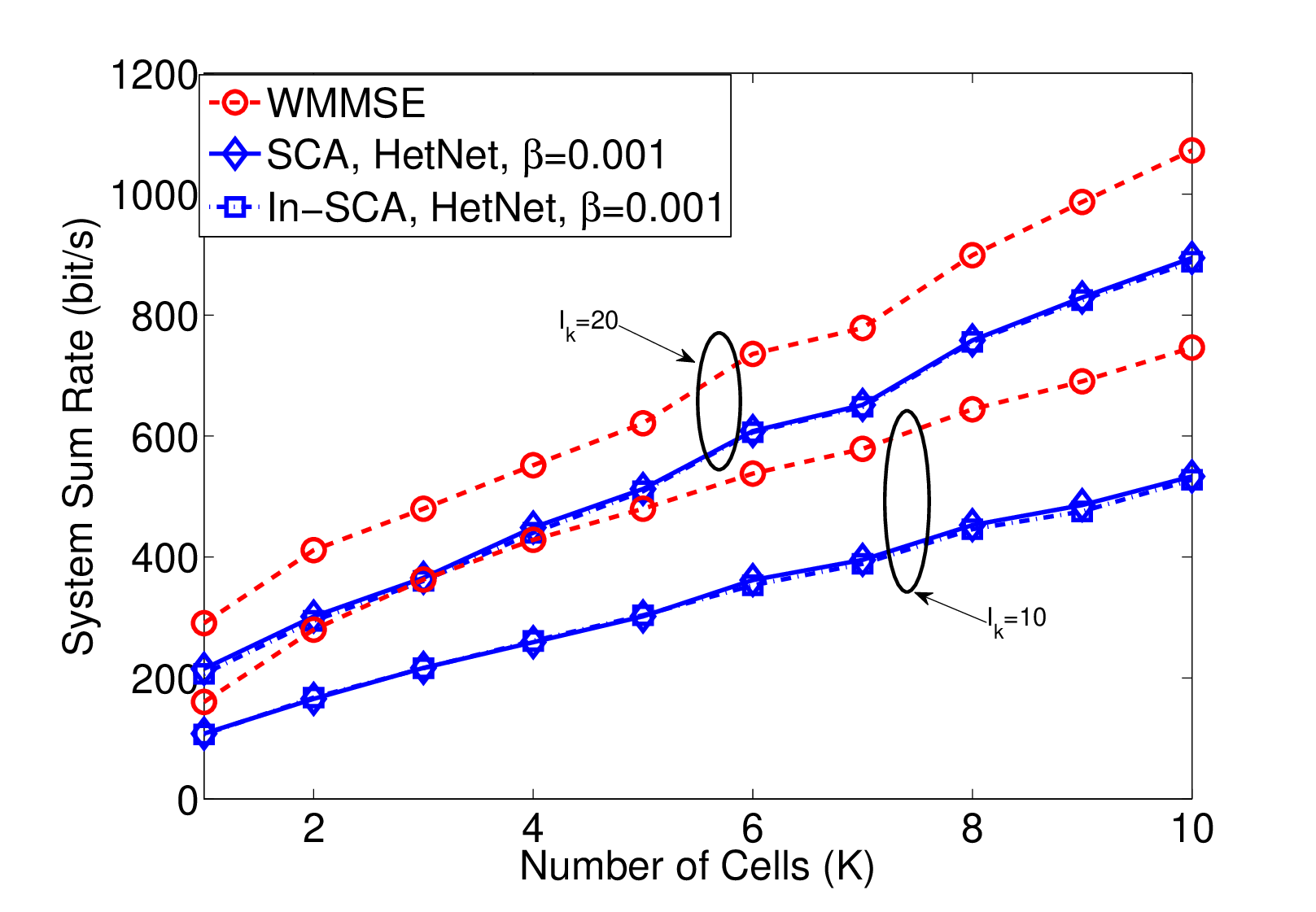}}}}
\hspace{0.3pc}
{\subfigure[][Averaged cluster size]{\resizebox{.48\textwidth}{!}{\includegraphics{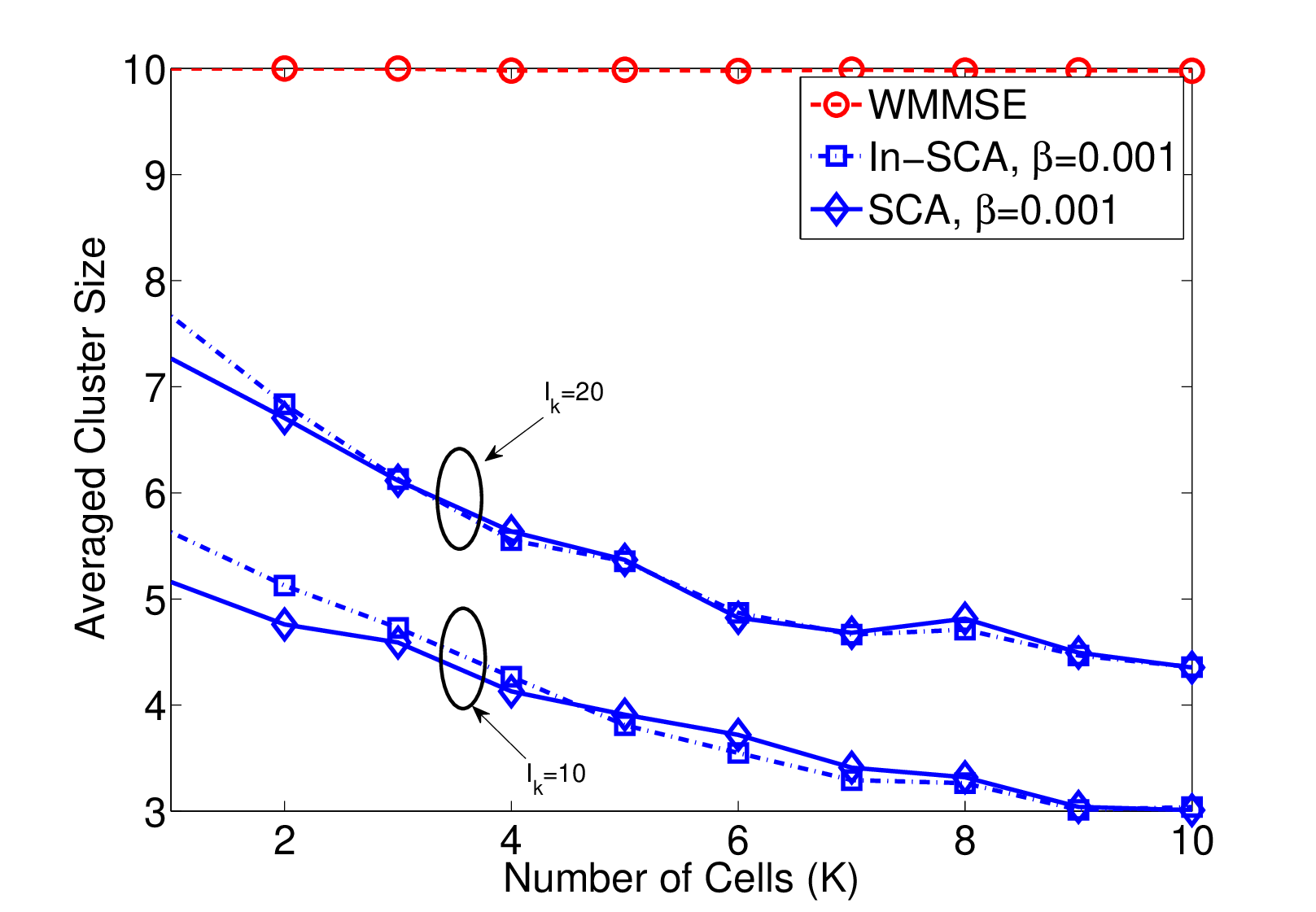}}}}
\end{center}
\vspace{-0.3cm}
\caption{ \footnotesize Comparison of the system performance of
different algorithms in HetNet. $K=[1,2,\cdots, 10]$, $P^{\rm tot}_k=20$dB, $|\cQ_k|=10$,
$|\cI_k|=[10,20]$, $M=4$, $N=2$, $d=1$ for all $i_k\in\cI$.
$\lambda=0.1$ when $|\cI_k|=10$, and $\lambda=0.05$ when
$|\cI_k|=20$. Left: System throughput. Right: Averaged cluster size.}\label{figRateSparse}
\vspace{-0.3cm}
 \end{figure*}

\begin{figure}
\includegraphics[width=
0.95\linewidth]{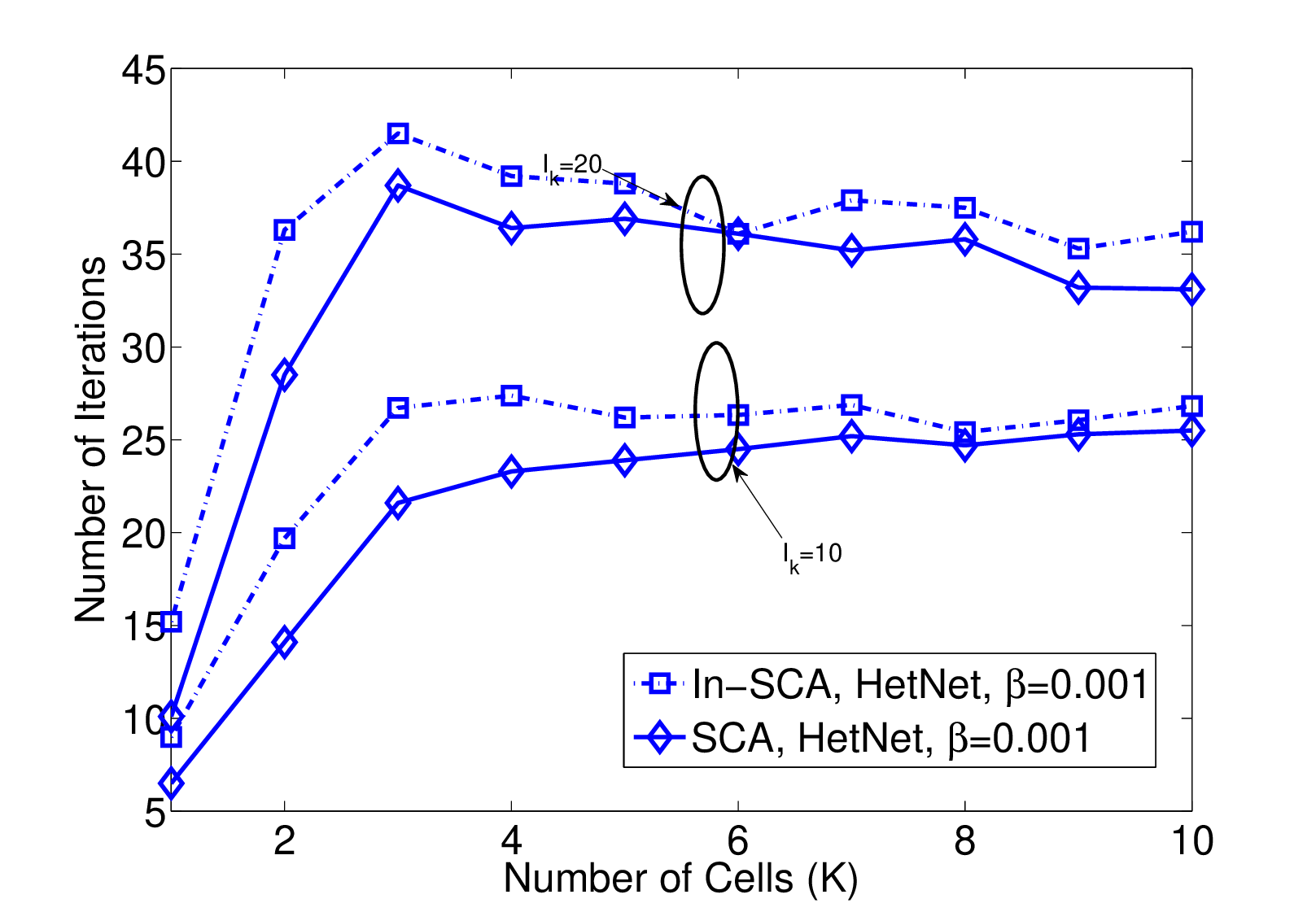}
\caption{\footnotesize Comparison of the system performance of
different algorithms in HetNet. $K=[1,2,\cdots, 10]$, $P^{\rm tot}_k=20$dB, $|\cQ_k|=10$,
$|\cI_k|=[10,20]$, $M=4$, $N=2$, $d=1$ for all $i_k\in\cI$.
$\lambda=0.1$ when $|\cI_k|=10$, and $\lambda=0.05$ when
$|\cI_k|=20$.}\label{figIterationHetNet}
\end{figure}

\begin{table}\small
\begin{tabular}{|c|c |c| c|c|c| } \hline
 & {\bf K=2}& {\bf K=4}
 &{\bf K=6} & {\bf K=8} & {\bf K=10}\\
\hline {\bf WMMSE} ($I_k=20$) &20.1 & 61.3
 & 73.4 & 89.9& 120.3\\
\hline {\bf SCA} ($I_k=20$) &105.2 & 204.4&258.1
 & 359.4& 443.7\\
\hline {\bf IN-SCA} ($I_k=20$) & 57.9 &88.2
 &109.9& 153.1& 189.1\\
 \hline
 \hline
{\bf WMMSE} ($I_k=10$) & 5.2& 15.2
 &31.2& 39.3& 58.1\\
\hline {\bf SCA} ($I_k=10$) & 13.9& 38.2&69.3& 140.5& 164.5\\
\hline {\bf IN-SCA} ($I_k=10$) & 7.4& 19.6
 &40.1& 45.3& 72.3\\
  \hline
\end{tabular}  \label{tableCPUTime}
\caption{CPU Time Needed for Different Algorithms
(Unit: Second)}
\vspace{-0.5cm}
\end{table}

\section{Conclusion}\label{secConclusion}

In this work we study an important family of interference management problems arising in the HetNet. The main novelty of this work lies in the proposal to achieve decomposition across the interference-coupled networks by using successive convex approximation. Our proposed approach has low computational complexity, as each of the subproblems to be solved is convex and decomposes across the cells. Depending on whether the subproblems are solved exactly, two general algorithms have been proposed, each having applications in a few interference management problems. In the future, we plan to extend our framework to other problems for interference management beyond those mentioned in this work.
\vspace{-0.3cm}

\appendix

\subsection{Proof of Lemma \ref{lemmaConvexity}}\label{app:Lemma1}
Consider the epigraph of $l_{i_k}({\bf V}_{i_k}, {\bf C}_{i_k})$,
i.e.,
$\left\{ ({\bf V}_{i_k}, {\bf C}_{i_k}, t) \mid l_{i_k}({\bf V}_{i_k}, {\bf C}_{i_k}) \leq t ~\right\}.$ It suffices to show that
the epigraph $({\bf V}_{i_k}, {\bf C}_{i_k}, t)$ is a convex set
\cite[Chapter 3]{boyd04}. To this end, let us consider the following
extended set (with ${\bf Z}_{i_k} \succeq {\bf 0}$ being a slack
variable):{\small
\begin{align} \label{eq:1}
   &\big\{ ({\bf V}_{i_k}, {\bf C}_{i_k},{\bf Z}_{i_k}, t)~|
   ~ {\rm Tr}[{\bf Z}_{i_k}] \leq t, ~{\bf Z}_{i_k} \succeq {\bf 0}, \nonumber\\
   &\quad {\bf Z}_{i_k} \succeq  \widehat{\bf E}_{i_k}^{-1/2} {\bf V}_{i_k}^H ({\bf H}^k_{i_k})^H {\bf C}_{i_k}^{-1}
   {\bf H}^k_{i_k} {\bf V}_{i_k} \widehat{\bf E}_{i_k}^{-1/2}  \big\}.
 \end{align}}
\!\!It is not hard to show that $({\bf V}_{i_k}, {\bf C}_{i_k}, t)$ is
just a projection of the set defined by~\eqref{eq:1}. Therefore, if
the extended set $({\bf V}_{i_k}, {\bf C}_{i_k},{\bf Z}_{i_k}, t)$
is convex, then $({\bf V}_{i_k}, {\bf C}_{i_k}, t)$ is also convex.
By applying Schur's complement, we have{\small
\begin{align}\label{eq:2}
 &{\bf Z}_{i_k} \succeq  \widehat{\bf E}_{i_k}^{-1/2} {\bf V}_{i_k}^H ({\bf H}^k_{i_k})^H
 {\bf C}_{i_k}^{-1} {\bf H}^k_{i_k} {\bf V}_{i_k} \widehat{\bf E}_{i_k}^{-1/2}\nonumber\\
 &\Longleftrightarrow
\begin{bmatrix}
  {\bf Z}_{i_k} &   \widehat{\bf E}_{i_k}^{-1/2} {\bf V}_{i_k}^H ({\bf H}^k_{i_k})^H \\
  {\bf H}^k_{i_k} {\bf V}_{i_k} \widehat{\bf E}_{i_k}^{-1/2} & {\bf C}_{i_k}
\end{bmatrix} \succeq {\bf 0} .
\end{align}}
\!\!This shows that \eqref{eq:1} is a convex set whenever $\bC_{i_k}\succeq 0$.

\subsection{Derivation of the equality in \eqref{eq:second_layer}} \label{appDirectional}
Let us investigate the directional derivative of the function $l_{i_k}$.
We can first obtain \eqref{eq:first_order}.
\begin{table*}[t]
\small{
{
\begin{align}\label{eq:first_order}
&\frac{d l_{i_k}(\bV_{i_k}+t\bd, \bC_{i_k}+t\bM_{i_k}) }{d
t}=\frac{d \trace\left[\widehat{\bE}_{i_k}^{-1} (\bV_{i_k}+t
\bd)^H(\bH^k_{i_k})^H(\bC_{i_k}+t\bM_{i_k})^{-1}\bH^k_{i_k}(\bV_{i_k}+t\bd)\right]}{d
t}\nonumber\\
&=\trace\left[\widehat{\bE}_{i_k}^{-1}
(\bd)^H(\bH^k_{i_k})^H(\bC_{i_k}+t\bM_{i_k})^{-1}\bH^k_{i_k}(\bV_{i_k}+t\bd)\right]+\trace\left[\widehat{\bE}_{i_k}^{-1} (\bV_{i_k}+t
\bd)^H(\bH^k_{i_k})^H(\bC_{i_k}+t\bM_{i_k})^{-1}\bH^k_{i_k}\bd\right]\nonumber\\
&\quad\quad-\trace\left[\widehat{\bE}_{i_k}^{-1} (\bV_{i_k}+t
\bd)^H(\bH^k_{i_k})^H(\bC_{i_k}+t\bM_{i_k})^{-1}\bM_{i_k}(\bC_{i_k}+t\bM_{i_k})^{-1}\bH^k_{i_k}(\bV_{i_k}+t\bd)\right].
\end{align}}
\rule{\linewidth}{0.2mm}}
\vspace{-0.5cm}
\end{table*}

Let $\bV_{i_k}=\bhV_{i_k}$, $\bd=\left(\bV_{i_k}-\bhV_{i_k}\right)$, $\bC_{i_k}=\bhC_{i_k}$ and $\bM_{i_k}=\mathbf{0}$, we obtain the expression{\small
\begin{align}
&\frac{d
l_{i_k}(\widehat{\bV}_{i_k}+t(\bV_{i_k}-\widehat{\bV}_{i_k}),
\widehat{\bC}_{i_k})}{d t}\bigg|_{t=0}\nonumber\\
&=\trace\left[\widehat{\bE}_{i_k}^{-1}(\bV_{i_k}-\widehat{\bV}_{i_k})^H(\bH^k_{i_k})^H\widehat{\bC}_{i_k}^{-1}
\bH^k_{i_k}\widehat{\bV}_{i_k}\right]\nonumber\\
&\quad+\trace\left[\widehat{\bE}_{i_k}^{-1}\widehat{\bV}^H_{i_k}(\bH^k_{i_k})^H\widehat{\bC}_{i_k}^{-1}
\bH^k_{i_k}(\bV_{i_k}-\widehat{\bV}_{i_k})\right].\nonumber
\end{align}}
Similarly, let { $\bV_{i_k}=\bhV_{i_k}$, $\bd=\mathbf{0}$, $\bC_{i_k}=\bhC_{i_k}$} and { $\bM_{i_k}=\left(\bC_{i_k}-\bhC_{i_k}\right)$}, we obtain the expression{\small
\begin{align}
&\frac{d l_{i_k}(\widehat{\bV}_{i_k},
\widehat{\bC}_{i_k}+t(\bC_{i_k}-\widehat{\bC}_{i_k}))}{d
t}\bigg|_{t=0}\nonumber\\
&=-\trace\left[\widehat{\bE}_{i_k}^{-1}\widehat{\bV}^H_{i_k}(\bH^k_{i_k})^H\widehat{\bC}_{i_k}^{-1}(\bC_{i_k}-\widehat{\bC}_{i_k})
\widehat{\bC}_{i_k}^{-1}\bH^k_{i_k}\widehat{\bV}_{i_k}\right]\nonumber.
\end{align}}
This completes the proof.

\vspace{-0.3cm}
\subsection{Proof of Theorem \ref{thmGeneralConvergence}}\label{appTheorem1}
We first prove the convergence for the SCA algorithm. It is easy to observe that the objective value of the
problem \eqref{problemSyS} is monotonically nondecreasing, i.e., we
have{\small
\begin{align}
u(\bV(t+1))&\stackrel{\rm (i)}\ge h^{\beta}(\bV(t+1);
\bV(t))-s(\bV(t+1))\nonumber\\
&\stackrel{\rm (ii)}\ge h^{\beta}(\bV(t);
\bV(t))-s(\bV(t))\stackrel{\rm
(iii)}=u(\bV(t))\label{eqSeriesInequality1}
\end{align}}
where step $\rm (i)$ is from \eqref{eqLowerBoundPropertyInequality};
step $\rm (ii)$ is from the fact that $\bV(t+1)$ is the solution to
the problem \eqref{problemLowerBoundGeneral}; step $\rm (iii)$ is
because of \eqref{eqLowerBoundPropertyInequality}. As both $f(\bV)$
and $s(\bV)$ are upper bounded for all $\bV$ in the feasible set, it
follows that the sequence $\{u(\bV(t))\}$ converges. Let $\bar{u}$
denote its limit. This result combined with
\eqref{eqSeriesInequality1} implies that{\small
\begin{align}
\lim_{t\to\infty} h^{\beta}(\bV(t+1);
\bV(t))-s(\bV(t+1))=\bar{u}.\label{eqHLimit1}
\end{align}}
\!\!Using a similar argument as in \cite[Lemma 1]{Razaviyayn12SUM}, and
use the fact that $f(\bullet)$ and $h(\bullet; \bhV)$ satisfy
\eqref{eqLowerBoundPropertyInequality},
we can show that for any feasible $\bhV$ the directional derivative
of $f(\bullet)$ at the point $\bhV$ equals that of $h^{\beta}(\bullet;\bhV)$ at
the point $\bhV$, i.e.,{\small
\begin{align}
\lim_{r\to 0}\frac{f(\bhV+r\bD)-f(\bhV)}{r}=\lim_{r\to
0}\frac{h^{\beta}(\bhV+r\bD;\bhV)-h^{\beta}(\bhV)}{r}\label{eqDerivitiveEqual}
\end{align}}
where $\bD$ is any feasible direction that satisfies $\bhV+\bD\in\mathcal{V}$. 

 Let $\{\bV(t_m)\}_{m=1}^{\infty}$ be a converging subsequence of
 $\bV(t)$ with limit $\bV^*$.  Clearly we have
{\[h^{\beta}(\bV^*; \bV^*)-s(\bV^*)=u(\bV^*)=\bar{u}.\]}
\!\!Then by Step 1 of the algorithm, the optimality of $\bV(t)$ implies
%
{\small\begin{align*}
&h^{\beta}(\bV(t_m+1); \bV(t_m))-s(\bV(t_m+1))\nonumber\\
&\ge h^{\beta}(\bV; \bV(t_m))-s(\bV),\;
\forall~ \bV\in\cV.
\end{align*}}
Taking limit on both sides and use \eqref{eqHLimit1}, we obtain
{\small$$\bar{u}\ge h^{\beta}(\bV; \bV^*)-s(\bV), \quad \forall~\bV.$$}
\!\!It follows that $h^{\beta}(\bV^*; \bV^*)-s(\bV^*)\ge h^{\beta}(\bV; \bV^*)-s(\bV)$
for all feasible $\bV$, which further implies that{\small
\begin{align*}
\left(h^{'}_{\bD}(\bV;\bV^*)-s^{'}_{\bD}(\bV)\right)\big|_{\bV=\bV^*}\le
0, \ \forall~\bD+\bV^*\in\cV.
\end{align*}}
Utilizing \eqref{eqDerivitiveEqual}, we obtain{\small
\begin{align*}
u^{'}_{\bD}(\bV)\big|_{\bV=\bV^*}&=\left(f^{'}_{\bD}(\bV)-s^{'}_{\bD}(\bV)\right)\big|_{\bV=\bV^*}\le
0, \forall~\bD+\bV^*\in\cV,
\end{align*}}
which says that $\bV^*$ is a stationary solution to the problem
\eqref{problemSyS}.

We then prove the result for the IN-SCA algorithm. Similarly as the previous case, we have{\small
\begin{align}
&u(\bV(t+1))\ge h^{\beta}(\bV(t+1);
\bV(t))-s(\bV(t+1))\nonumber\\
&\stackrel{\rm (i)}\ge h^{\beta}(\bV(t);
\bV(t))-s(\bV(t))+\eta\|\bV(t+1)-\bV(t)\|^2_F\nonumber\\
&=u(\bV(t))+\eta\|\bV(t+1)-\bV(t)\|^2_F\label{eqSeriesInequality2}
\end{align}}
\!\!where $\rm (i)$ is due to condition \eqref{eq:condition_inexact}. The above inequality implies that{\small
\begin{align*}
u(\bV(t+1))-u(\bV(0))\ge \sum_{t=1}^{t}\eta\|\bV(t+1)-\bV(t)\|^2_F.
\end{align*}}
\!\!Taking limit on both sides, and by the upper-boundedness of $u(\bV)$ for all $\bV\in\cV$, we obtain{\small
\begin{align}
\lim_{t\to\infty}&\|\bV(t+1)-\bV(t)\|^2_F=0,\quad \lim_{t\to\infty} u(\bV(t+1))=\bar{u}.\label{eqHLimit2}
\end{align}}
\!\!Let $\{\bV(t_m)\}_{m=1}^{\infty}$ be a converging subsequence of
 $\bV(t)$, and denote its limit as $\bV^*$. Combining the first equation in \eqref{eqHLimit2} and the second condition in \eqref{eq:condition_inexact}, we have{\small
 \begin{align*}
 &\bV^*=\arg\min_{\bV} h^{\beta}(\bV; \bV^*)-s(\bV) \nonumber\\
 &\Longleftrightarrow
 h^{\beta}(\bV^*; \bV^*)-s(\bV^*)\ge  h^{\beta}(\bV; \bV^*)-s(\bV), \; \forall~\bV\in\cV.
 \end{align*}}
\!\!Then the conclusion follows from a similar argument as in the first part of the proof.

\vspace{-0.2cm}
\subsection{Proof of Corollary \ref{cor:equivalence}}\label{eq:appC}
We focus on the case that each coordinate $m_k\in\cQ_k$ is updated precisely once in Step 3. The case in which each coordinate is updated {\it at least} once is a straightforward extension.

We only need to verify the two conditions \eqref{eq:condition_inexact}. Define the following short-handed notation{\small
\begin{align}
\bV^{-m_k}_{i_k}&:=\{\bV^{p_k}_{i_k}\}_{p_k\ne m_k}, \; \bV^{-m_k}:=\{\bV^{p_k}\}_{p_k\ne m_k}\nonumber.
\end{align}}
\!\!Also define $\nabla_{m_k} g^{\beta}_{i_k}(\bV^{m_k}_{i_k}, \bV^{-m_k}_{i_k};\bhV)$ as the gradient of $g^{\beta}_{i_k}(\bullet)$ with respect to the block $\bV^{m_k}_{i_k}$.

At any given $\bhV$, we fix $\bV^{-m_k}$ and update $\bV^{m_k}$ by using \eqref{eqVIBCComputeV}. First observe that the optimality condition for the per-block problem \eqref{problemVIBC-BLK} is given by{\small
\begin{align}\label{eq:optimality}
&\sum_{i_k\in \cI_k}\left\langle \nabla_{m_k}g^{\beta}_{i_k}\left(\bV^{m_k*}_{i_k}, \bV^{-m_k}_{i_k}; \bhV\right), \bV^{m*}_{i_k}-\btV^{m_k}_{i_k}\right\rangle\ge 0,\nonumber\\
& \quad\quad\quad\quad\quad\forall\;\btV^{m_k}\in \cV^{m_k}.
\end{align}}
Utilizing this property, we have{\small
\begin{align}
&\sum_{i_k\in \cI_k} g^{\beta}_{i_k}\left(\bV^{m_k*}_{i_k}, \bV^{-m_k}_{i_k}; \bhV\right)-g^{\beta}_{i_k}\left(\bhV^{m_k}_{i_k}, \bV^{-m_k}_{i_k}; \bhV\right)\label{eq:strong:convexity}\\
&\stackrel{\rm (i)}\ge \sum_{i_k\in \cI_k}\left\langle \nabla_{m_k}g^{\beta}_{i_k}\left(\bV^{m*}_{i_k}, \bV^{-m_k}_{i_k}; \bhV\right), \bV^{m*}_{i_k}-\bhV^{m_k}_{i_k}\right\rangle\nonumber\\
&\quad +\sum_{i_k\in\cI_k}\frac{\beta}{2}\|\bV^{m_k*}_{i_k}-\bhV^{m_k}_{i_k}\|^2_F\nonumber\\
&\stackrel{\rm (ii)}\ge \sum_{i_k\in\cI_k}\frac{\beta}{2}\|\bV^{m_k*}_{i_k}-\bhV^{m_k}_{i_k}\|^2_F\nonumber
\end{align}}
\!\!where $\rm (i)$ is from the strong concavity of {$g^{\beta}_{i_k}\left(\bV^{m_k}_{i_k}, \bV^{-m_k}_{i_k}; \bhV\right)$} with respect to $\bV^{m_k}_{i_k}$ (with modulus $\beta$); 
$\rm (ii)$ comes from the optimality condition \eqref{eq:optimality}. After one round of the update where each $m_k$ is updated once, we have{\small
\begin{align}
&h^{\beta}(\bV^*; \bhV)-h^{\beta}(\bhV; \bhV)\nonumber\\
&\ge \sum_{k\in\cK}\sum_{i_k\in \cI_k}\sum_{m_k\in \cQ_k}\frac{\beta}{2}\|\bV^{m_k*}_{i_k}-\bhV^{m_k}_{i_k}\|^2_F\ge \frac{\beta}{2}\|\bV^*-\bhV\|^2_F\nonumber
\end{align}}
\!\!which proves the first condition in \eqref{eq:condition_inexact}.

To show the second condition in \eqref{eq:condition_inexact}, we fix $\bV^{-m_k}=\bhV^{-m_k}$ and consider $\bV^{m_k*}$ obtained by using \eqref{eqVIBCComputeV}. Clearly if for each $m_k\in\cQ_k$, $\bhV^{m_k}=\bV^{m_k*}$ is true, then we have{\small
\begin{align*}
\bhV^{m_k}=\arg \max_{\bV^{m_k}}&\sum_{i_k\in{\mathcal{I}_k}}g^{\beta}_{i_k}(\bV^{m_k}_{i_k},\bhV^{-m_k}_{i_k};\bhV), \; \forall~m_k\in\cQ_k.
\end{align*}}
\!\!The optimality conditions of these $|\cQ_k|$ problems are given below
{\small
\begin{align}\label{eq:optimality2}
&\sum_{i_k\in \cI_k}\left\langle \nabla_{m_k}g^{\beta}_{i_k}\left(\bhV^{m_k}_{i_k}, \bhV^{-m_k}_{i_k}; \bhV\right), \bhV^{m_k}_{i_k}-\btV^{m_k}_{i_k}\right\rangle\ge 0, \nonumber\\
&\quad \quad\quad\quad\; \forall~\btV^{m_k}\in \cV^{m_k}, \; \forall~m_k\in \cQ_k.
\end{align}}
which collectively imply the optimality condition of the following problem {\small
\begin{align*}
\bhV^{k}=\arg \max_{\bV^{k}}\;&\sum_{i_k\in{\mathcal{I}_k}}g^{\beta}_{i_k}(\bV_{i_k};\bhV).
\end{align*}}
Enumerating over all cells $k\in\cK$, we obtain{\small
\begin{align*}
\bhV&=\arg \max_{\bV}\;\sum_{k\in\cK}\sum_{i_k\in{\mathcal{I}_k}}g^{\beta}_{i_k}(\bV_{i_k}; \bhV)=\arg\max_{\bV}\; h^{\beta}(\bV; \bhV).
\end{align*}}
\!\!which satisfies \eqref{eq:condition_inexact2}. The proof is completed.}

\bibliographystyle{IEEEbib}
{\scriptsize
\bibliography{ref}
}

\end{document}